%% file: main.tex
\documentclass[sigconf, nonacm, pdfa]{acmart}

\newcommand\vldbdoi{10.14778/3819518.3819545}
\newcommand\vldbpages{2210-2223}
\newcommand\vldbvolume{19}
\newcommand\vldbissue{9}
\newcommand\vldbyear{2026}
\newcommand\vldbauthors{\authors}
\newcommand\vldbtitle{\shorttitle} 
\newcommand\vldbavailabilityurl{https://github.com/Geonho-Lee/SafeQL}
\newcommand\vldbpagestyle{empty}

\usepackage{enumitem}
\usepackage{multirow}
\usepackage{makecell}
\usepackage{mathtools}
\usepackage{setspace}
\usepackage[linesnumbered,ruled,noend]{algorithm2e}
\usepackage{xcolor}
\usepackage{framed}
\usepackage{multicol}
\usepackage{balance}
\usepackage[utf8]{inputenc}
\usepackage{amsthm}
\usepackage{amsmath}
\usepackage[a-2b]{pdfx}

\definecolor{notice}{RGB}{0, 0, 0}

\definecolor{myred}{RGB}{255, 0, 0}
\definecolor{mysafety}{RGB}{255, 221, 221}
\definecolor{myorange}{RGB}{244, 177, 131}
\definecolor{myblue}{RGB}{157, 195, 230}
\definecolor{mygreen}{RGB}{169, 209, 142}

\usepackage[normalem]{ulem}
\usepackage{makecell}

\newcolumntype{P}[1]{>{\arraybackslash}p{#1}}
\newcolumntype{M}[1]{>{\centering\arraybackslash}m{#1}}

\SetCommentSty{mycommfont}

\SetFuncSty{myfuncfont}

\newcommand{\tightarrow}[1]{%
  \raisebox{-0.2ex}{$\;\mathrel{\xrightarrow{\text{\raisebox{-0.8ex}{\scriptsize $#1$}}}}\;$}%
}

\DeclareMathOperator*{\argminA}{arg\,min}

\begin{document}

\title{\textcolor{notice}{SafeQL: Search-based Refinement for Safe and Efficient LLM-based Text-to-SQL}}

\author{Geonho Lee}
\email{ghlee5084@kaist.ac.kr}
\affiliation{%
  \institution{Korea Advanced Institute of Science and Technology}
  \country{Republic of Korea}
}

\author{Min-Soo Kim}
\authornote{Corresponding author.}
\email{minsoo.k@kaist.ac.kr}
\affiliation{%
  \institution{Korea Advanced Institute of Science and Technology}
  \country{Republic of Korea}
}

\begin{abstract}
\textcolor{notice}
{
Large language models (LLMs) have advanced Text-to-SQL by enabling natural language interfaces to databases without task-specific fine-tuning.
However, existing LLM-based systems remain unreliable, often generating SQL queries that are invalid under the database schema, referencing non-existent tables, attributes, functions, or values.
Such errors persist because interactions with the database management system (DBMS) are typically limited to error messages, leaving it in a largely passive role during query refinement.
This paper proposes SafeQL, \textit{a search-based refinement paradigm that redefines the role of the DBMS as an active guide in the refinement process}.
Instead of regenerating entire queries after execution failure, SafeQL interprets DBMS feedback to incrementally repair only the erroneous components.
Each refinement step is formulated as a guided search within a \textit{safe query space}, where candidate queries are progressively validated through DBMS execution, thereby converging to an executable query and preventing repeated regeneration of errors.
Experiments on the Bird and Spider benchmarks show that SafeQL significantly improves execution accuracy and efficiency compared to regeneration-based methods.
}
\end{abstract}

\maketitle

\pagestyle{\vldbpagestyle}
\begingroup\small\noindent\raggedright\textbf{PVLDB Reference Format:}\\
\vldbauthors. \vldbtitle. PVLDB, \vldbvolume(\vldbissue): \vldbpages, \vldbyear.\\
\href{https://doi.org/\vldbdoi}{doi:\vldbdoi}
\endgroup
\begingroup
\renewcommand\thefootnote{}\footnote{\noindent
This work is licensed under the Creative Commons BY-NC-ND 4.0 International License. Visit \url{https://creativecommons.org/licenses/by-nc-nd/4.0/} to view a copy of this license. For any use beyond those covered by this license, obtain permission by emailing \href{mailto:info@vldb.org}{info@vldb.org}. Copyright is held by the owner/author(s). Publication rights licensed to the VLDB Endowment. \\
\raggedright Proceedings of the VLDB Endowment, Vol. \vldbvolume, No. \vldbissue\ %
ISSN 2150-8097. \\
\href{https://doi.org/\vldbdoi}{doi:\vldbdoi} \\
}\addtocounter{footnote}{-1}\endgroup

\ifdefempty{\vldbavailabilityurl}{}{
\vspace{.3cm}
\begingroup\small\noindent\raggedright\textbf{PVLDB Artifact Availability:}\\
The source code, data, and/or other artifacts have been made available at \url{\vldbavailabilityurl}.
\endgroup
}

\input{sec_1}
\input{sec_2}
\input{sec_3}
\input{sec_4}
\input{sec_5}
\input{sec_6}
\input{sec_7}
\input{sec_8}
\input{sec_9}

\vspace{-3mm}
\begin{acks}
\vspace{-1mm}
This work was supported by the National Research Foundation of Korea(NRF) grant funded by the Korea government(MSIT)(No. RS-2024-00347471) and Institute of Information \& communications Technology Planning \& Evaluation (IITP) grant funded by the Korea government(MSIT) (No. 2019-0-01267, GPU-based Ultrafast Multi-type Graph Database Engine SW).
\end{acks}

\bibliographystyle{ACM-Reference-Format}
\bibliography{main}

\end{document}

%% file: sec_1.tex
\section{Introduction}
\label{sec:introduction}

Natural language interfaces to databases (NLIDBs) have long been a key research goal at the intersection of database systems and natural language processing~\cite{zelle1996learning, iacob2020neural}.  
Among them, Text-to-SQL—the task of translating a natural language question into an executable SQL query—represents the most concrete realization of this vision, enabling users to access structured data through natural language rather than query syntax.  
By eliminating the need for SQL expertise, Text-to-SQL systems make data access more inclusive and have broad applications in analytics, business intelligence, and scientific data exploration~\cite{kim2020natural, katsogiannis2023survey, liu2025survey}.

Research on Text-to-SQL has evolved through several generations.  
Early symbolic systems such as PRECISE~\cite{popescu_towards_nodate}, NaLIR~\cite{li_nalir_2014}, SQLizer~\cite{yaghmazadeh_sqlizer_2017}, and ATHENA~\cite{saha_athena_2016} relied on rule-based parsing and semantic grammars that explicitly mapped natural language to SQL templates. 
Although interpretable, these systems were fragile and required manual engineering for each domain.  
Subsequent neural semantic parsers replaced handcrafted rules with data-driven models trained on paired natural-language and SQL examples.  
Representative models include Seq2SQL~\cite{zhong_seq2sql_2017}, SQLNet~\cite{xu_sqlnet_2017}, IRNet~\cite{guo_towards_2019}, RAT-SQL~\cite{wang_rat-sql_2020}, and SmBoP~\cite{rubin_smbop_2021}, which introduced neural encoders and schema linking for cross-domain generalization.  
Later work incorporated syntax constraints and intermediate representations to improve validity and cross-schema transfer~\cite{lin_bridging_2020, scholak_picard_2021}.  
Despite substantial progress, these models still require expensive annotation and retraining whenever the database schemas change, limiting their scalability in practical deployments.

The emergence of large language models (LLMs) such as GPT~\cite{brown_language_2020} has profoundly transformed the Text-to-SQL landscape.  
Prompt-based approaches like DAIL-SQL~\cite{gao_text--sql_2024} and MCS-SQL~\cite{lee_mcs-sql_2025} demonstrated that general-purpose LLMs can perform Text-to-SQL generation through schema-aware prompting, without task-specific fine-tuning.
Building on this paradigm, subsequent prompting techniques—including few-shot sampling~\cite{brown_language_2020, li2024pet, trummer2022codexdb, nan_enhancing_2023, gao_text--sql_2024, lee_mcs-sql_2025}, chain-of-thought reasoning~\cite{zhang2023act, wei_chain--thought_2022}, and least-to-most decomposition~\cite{pourreza2024dts, pourreza_din-sql_2023, zhou_least--most_2022}—further improved logical coherence and interpretability.  
More recently, multi-agent systems including MAC-SQL~\cite{wang2025mac}, CHESS-SQL~\cite{talaei_chess_2024}, CHASE-SQL~\cite{pourreza_chase-sql_2024}, Alpha-SQL~\cite{li_alpha-sql_2025}, and OpenSearch-SQL~\cite{xie_opensearch-sql_2025} have introduced collaborative reasoning, in which multiple LLM agents decompose, verify, and align candidate SQL queries to ensure consistent output. 
These systems mark a shift from one-shot generation to autonomous reasoning pipelines that incorporate coordination among multiple agents.

\textcolor{notice}{
Despite these advances, reliability remains a critical challenge due to the inherent hallucination tendencies of LLMs.  
In practice, many failures arise when generated queries cannot be executed, for example, due to references to nonexistent relations, attributes, or functions. 
Even state-of-the-art LLMs often produce queries that appear plausible but fail at the DBMS level. 
Such failures expose a mismatch between the user’s intent and the database schema and disrupt downstream workflows, making it essential to refine queries into \textit{safe (i.e., executable)} ones.}

Recent studies~\cite{pourreza_din-sql_2023, wang2025mac, talaei_chess_2024, ren_power_2025, xie_opensearch-sql_2025} have attempted to mitigate this issue through \textit{regeneration-based refinement}: executing the generated SQL, capturing DBMS error messages, and prompting the LLM to regenerate a corrected query.
However, this paradigm treats the DBMS primarily as a \textit{passive checker} that provides error messages, and repeatedly regenerates entire queries based on the errors.
As a result, the regeneration loop is both \textit{inefficient}—requiring excessive tokens and computation—and \textit{unreliable}, often reintroducing previous errors with no guarantee of convergence.
Figure \ref{fig:high_level_idea}(a) illustrates a regeneration-based refinement process.
The model begins with an initial query $q_0$, which fails during execution.
Upon failure, the entire query is discarded, and a new query $q_1$ is regenerated based on the previous error message.
This cycle repeats (e.g., $q_2$, $q_3$ …) until an executable query is obtained.
Each iteration re-parses and re-executes the entire query, resulting in redundant computation and limited reuse of valid fragments.

In contrast, we propose SafeQL, a refinement framework that achieves safe and efficient correction 
through a \textit{search-based refinement} paradigm that fundamentally rethinks how Text-to-SQL errors are corrected. 
Rather than discarding the initial query $q_0$, SafeQL interprets and incrementally refines it.
As shown in Figure~\ref{fig:high_level_idea}(b), when $q_0$ fails, SafeQL does not regenerate a new $q_1$.
Instead, it analyzes the DBMS feedback to precisely locate the faulty component—whether a relation, attribute, or function—and applies a minimal, structure-preserving correction as depicted in blue arrows in the Figure~\ref{fig:high_level_idea}(b).
This transforms $q_0$ directly into refined and executable queries, while preserving all valid fragments of the original logic.
Through this process, SafeQL replaces the traditional trial-and-error regeneration loop with a guided search over a \textit{safe query space} as depicted in the green part of the Figure~\ref{fig:high_level_idea}(b), where candidates are incrementally repaired to satisfy schema-level constraints enforced by the DBMS.
Since this safe query space may contain multiple executable refinements for a single input $q_0$—for example, when an unknown attribute can be repaired using several attributes in the schema-—SafeQL selects the candidate that is semantically the most aligned with the initial query $q_0$.
This paradigm shift redefines the DBMS from a passive checker into an \textit{active reasoning engine} that incrementally narrows the search to the most faithful executable query, improving both safety and efficiency.

\vspace{-1mm}
\begin{figure}[hbtp]
    \centerline{\includegraphics[width=3.4in]{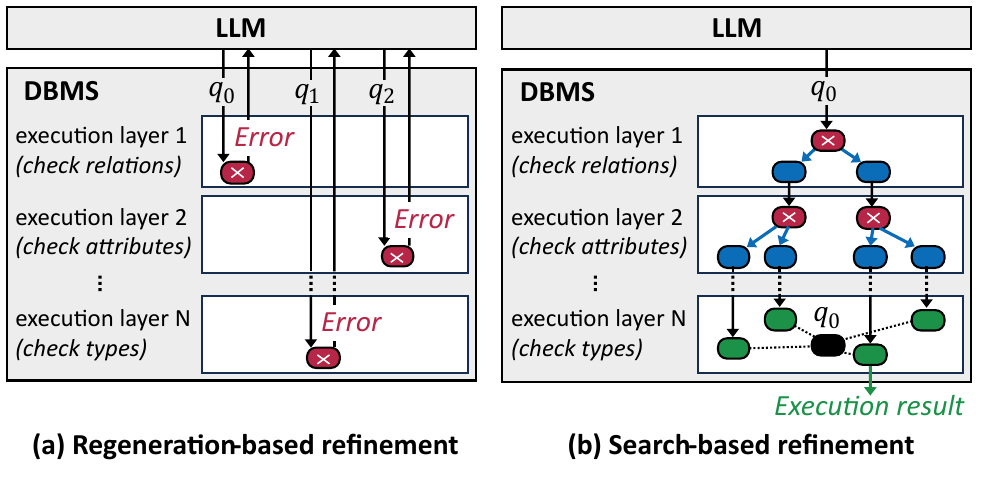}}
    \vspace{-5mm}
    \caption{High-level comparison of refinement paradigms.}
    \label{fig:high_level_idea}
\end{figure}
\vspace{-3mm}

Beyond the paradigm itself, the effectiveness of SafeQL critically depends on how search and optimization processes are conducted.  
The safe query space can grow exponentially, as each detected error may yield multiple candidate substitutions for relations, attributes, or values.
To manage this combinatorial complexity, SafeQL employs a best-first search strategy guided by a \textit{semantic distance metric} that combines structural similarity and embedding-based relevance.  
This prioritization ensures that the most promising candidates are explored first, improving convergence while preserving the user's original query intent.  
In terms of optimizations, SafeQL incorporates \textit{type-based pruning}, which eliminates type-inconsistent refinements, and \textit{top-$K$ pruning}, which restricts candidate substitutions to the most semantically relevant ones.
Together, these optimizations make search-based refinement both practical and scalable.

Equally important is how SafeQL is realized as a practical, system-level framework rather than merely an algorithmic concept.
SafeQL is integrated directly into the DBMS, enabling refinement through the database’s own parser, binder, and type analyzer instead of treating SQL as plain text.
This tight integration allows for precise error localization, while system-level optimizations—such as caching and vector indexing—make semantic similarity computations efficient at scale.
Moreover, SafeQL incorporates a \textit{hybrid refinement strategy} that combines its search-based refinement with selective regeneration to handle cases where the initial query is fundamentally misgenerated.
Together, these system-level integrations make our SafeQL a complete, end-to-end framework that delivers reliable and efficient Text-to-SQL refinement in real-world environments.

\noindent
Our main contributions are summarized as follows: \vspace{-1mm}
\begin{itemize}[labelindent=0.0em,labelsep=0.5em,leftmargin=*]

    \item We propose a search-based refinement paradigm that reformulates Text-to-SQL refinement as a guided exploration of a safe query space, where execution feedback serves as actionable semantic guidance rather than passive error reporting.

    \item We propose a best-first search guided by a semantic distance metric, augmented with type-based and top-K pruning, enabling scalable and accurate refinement.

    \item We implement SafeQL as a DBMS-integrated refinement framework, incorporating caching, vector indexing, and hybrid refinement to achieve efficient and robust execution.

    \item We conduct comprehensive experiments on Bird and Spider benchmarks, demonstrating up to 5.8\% improvements in execution accuracy and 15$\times$ reductions in token usage compared to regeneration-based baselines.

\end{itemize}

The remainder of this paper is organized as follows.
Section~\ref{sec:background} provides background and definitions.
Section~\ref{sec:search_space} introduces the safe query space and the refinement tree formulation.
Section~\ref{sec:search_algorithm} details the search strategy, and Section~\ref{sec:pruning_methods} presents pruning optimizations for scalable refinement.
Section~\ref{sec:system} describes the system architecture, and Section~\ref{sec:evaluation} reports the experimental results.
Finally, Section~\ref{sec:related_work} discusses related work, and Section~\ref{sec:conclusion} concludes the paper.

\vspace{-2mm}

%% file: sec_2.tex
\section{Preliminaries}
\label{sec:background}

\subsection{Text-to-SQL Stages}
\label{sec:background:text-to-sql}

The Text-to-SQL task translates a natural language question into an executable SQL query over a given database.
Typically, this process can be decomposed into three stages:
(1) extracting database information, 
(2) generating candidate SQL queries, and 
(3) refining them using DBMS execution.
We formalize these stages as follows.

\begin{definition}[Text-to-SQL Process]
\label{def:text2sql}
Let $D$ denote a database, and let $\mathcal{N}$, $\mathcal{P}$, and $\mathcal{Q}$ denote the
natural language, prompt, and query spaces, respectively.
The Text-to-SQL process is defined as
\[
\mathcal{N}
\xrightarrow{\,f_{\text{extract}}\,}
\mathcal{P}
\xrightarrow{\,f_{\text{generate}}\,}
\mathcal{Q}
\xrightarrow{\,f_{\text{refine}}\,}
\mathcal{Q},
\text{\;\;\;\;where: }
\]
\begin{itemize}[labelindent=0em,labelsep=0.5em,leftmargin=*]
    \item 
    \textbf{Extraction.}
    \[
    f_{\text{extract}} : 
    \textcolor{notice}{\mathcal{N} \times D \to \mathcal{P}}
    \]
    interprets the natural language question in the context of $D$ and constructs a prompt enriched with database-level hints.

    \item 
    \textbf{Generation.}
    \[
    f_{\text{generate}} : 
    \mathcal{P} \to \mathcal{Q}
    \]
    synthesizes an SQL query using the LLM that reflects the intent expressed in the prompt.

    \item 
    \textbf{Refinement.}
    \[
    f_{\text{refine}} : 
    \textcolor{notice}{\mathcal{Q} \times D \to \mathcal{Q}}
    \]
    corrects the generated query using feedback obtained from DBMS execution over $D$. \hfill $\Box$
\end{itemize}
\end{definition}

Although LLMs can generate plausible SQL queries, they may still suffer from hallucinations.
Hence, the refinement stage using DBMS execution is crucial to ensure reliability, regardless of whether the generation is based on the prompt~\cite{gao_text--sql_2024, lee_mcs-sql_2025, li2024pet, pourreza_din-sql_2023, zhang2023act} or the agent~\cite{dong2023c3, talaei_chess_2024, sheng_csc-sql_2025, li_alpha-sql_2025, xie_opensearch-sql_2025}. 

\subsection{Refinement Stage}
\label{sec:background:refinement}

In the refinement stage, the execution of the DBMS plays a central role in guiding the correction of the query, as shown in Figure~\ref{fig:high_level_idea}.  
Execution errors provide essential feedback signals that enable the Text-to-SQL system
to identify and correct invalid parts of a query.  
By interpreting these errors, the system can transform an invalid query
into one that is both executable and semantically faithful to the user’s intent.
To formalize how DBMS execution supports this process,
we first define the notions of \textit{Database} and \textit{Execution Error}, describing how DBMS is utilized during refinement.
Here, we adopt a simplified database definition tailored for Text-to-SQL scenarios (e.g., assuming that all attributes in $Atts$ are distinct).

\begin{definition}[Database] 
\label{def:database-schema}
A database $D$ is a collection of relations, attributes, and values under schema and type constraints. 
Formally,
\[
D = (Rels, Atts, Vals, Typs, \Lambda, \Gamma) \text{, where:} 
\]
\begin{itemize}[labelindent=0.0em,labelsep=0.5em,leftmargin=*]
    \item $Rels$ and $Atts$ are the sets of relations and attributes, respectively.
    \item $Vals$ is a set of instance-level values that populate the relations.
    \item $Typs$ is a set of data types (e.g., int, varchar, text).
    \item $\Lambda: Rels \rightarrow 2^{Atts}$ is a \textit{relation signature} mapping each relation to its corresponding set of attributes.
    \item $\Gamma: Atts \cup Vals \rightarrow Typs$ is a \textit{type constraint} mapping each attribute or value to its data type.
    \hfill $\Box$
\end{itemize}
\end{definition}

Figure~\ref{fig:example-data} illustrates a simple schema with two relations, \textsf{Users} and \textsf{Posts}, 
where $\Lambda(\textsf{Users}) = \{\textsf{ID}, \textsf{name}, \textsf{age}\}$ and 
$\Gamma(\textsf{name}) = \textsf{varchar}$.

\begin{figure}[hbtp]
    \vspace{-3mm}
    \centerline{\includegraphics[width=3.3in]{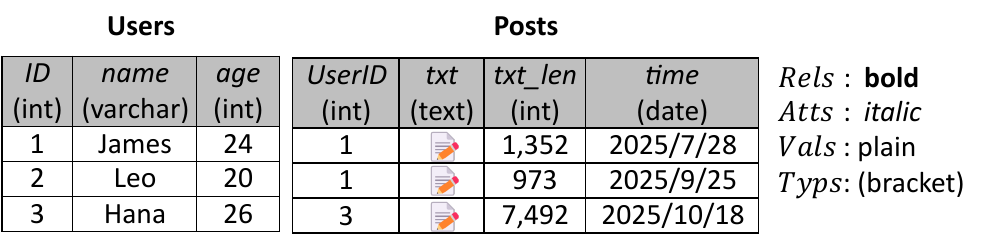}}
    \vspace{-4mm}
    \caption{Example of a database.}
    \vspace{-3mm}
    \label{fig:example-data}
\end{figure}

During query refinement, some candidate queries fail during execution due to schema or type violations.
We refer to these cases as \textit{execution errors}, as formally defined as follows.

\begin{definition}[Execution Error]
\label{def:execution-error}
Consider a DBMS that manages a database $D$. We denote the execution of $q \in \mathcal{Q}$ by the DBMS as $D \;\vdash\; q \;\Rightarrow\; \llbracket q \rrbracket$, where DBMS executes $q$ on $D$ and returns the result $\llbracket q \rrbracket$.  
An \emph{execution error} occurs when the execution fails, in which case we write: 
\[
\mathsf{DBMS} : D \;\vdash\; q \;\Rightarrow\; \epsilon(q), 
\]
where $\epsilon(q)$ is the error reported by the DBMS.
The most common errors in the Text-to-SQL scenario include:

\begin{itemize}[labelindent=0em,labelsep=0.5em,leftmargin=*,itemsep=0em]
    \item \textbf{\textit{unknown relation}:} The query references a non-existent table.
    \item \textbf{\textit{unknown attribute}:} The query refers to an undefined column.
    \item \textbf{\textit{unknown function (or operator)}:} The query invokes a function (or operator) not supported by the DBMS.
    \item \textbf{\textit{empty result}:} The query executes successfully, but returns no tuples due to unsatisfiable predicates.
    \hfill $\Box$
\end{itemize}
\vspace{-1mm}
\end{definition}

Figure~\ref{fig:example-errors} presents typical examples of execution errors along with their correct (gold) SQL counterparts.  
The top-left query fails because it references a non-existent attribute (\textsf{username}), whereas the top-right query accesses attributes (\textsf{txt}, \textsf{time}) from another relation without performing a join.  
In the bottom-left example, the function \textsf{strftime} fails to execute, illustrating how dialect-specific differences in SQL function support can cause errors.
Finally, the bottom-right query returns an empty result due to a case-sensitive predicate mismatch (\textsf{`leo'} vs. \textsf{`Leo'}).  
\textcolor{notice}
{
Although empty results do not trigger runtime errors in standard DBMSs,
many Text-to-SQL benchmarks treat them as errors~\cite{shen2025study, ren_power_2025}.
We follow this convention, while noting that empty results may not always correspond to incorrect queries in practice.
Our system also provides configurable controls for refinement per error type, and reports detailed error statistics (Figure~\ref{fig:error-statistics}) to facilitate analysis.
}

\vspace{-4mm}
\begin{figure}[hbtp]
    \centerline{\includegraphics[width=2.9in]{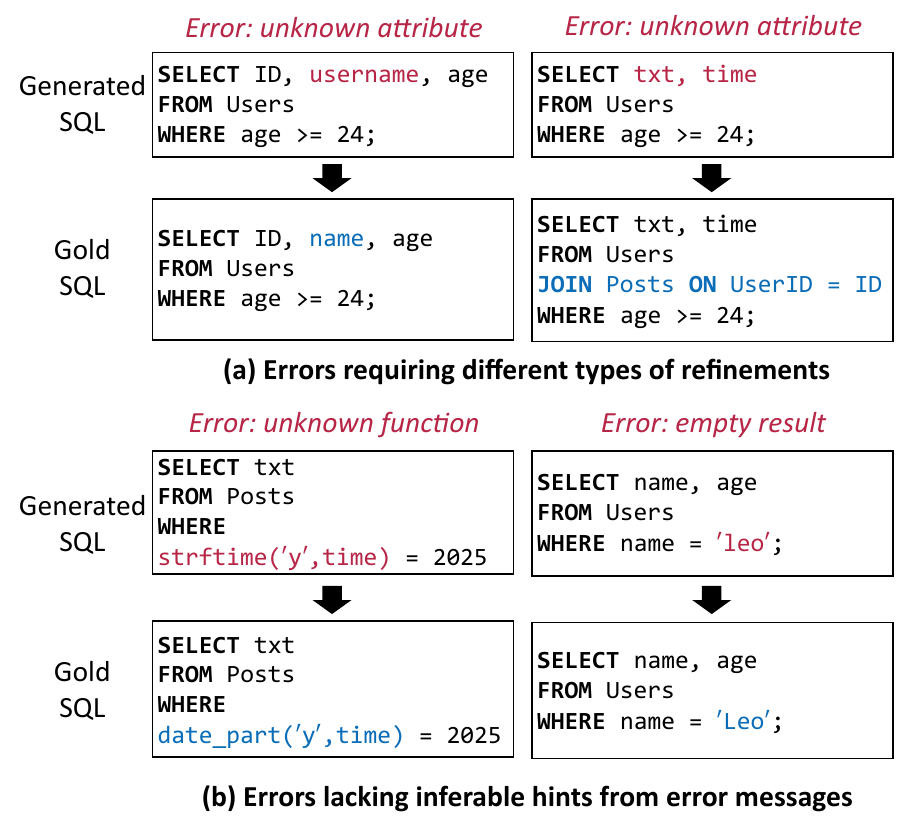}}
    \vspace{-5mm}
    \caption{Examples of erroneous SQLs and gold SQLs.}
    \vspace{-4mm}
    \label{fig:example-errors}
\end{figure}

Regeneration-based Text-to-SQL methods~\cite{pourreza_din-sql_2023, wang2025mac, ren_power_2025, talaei_chess_2024, xie_opensearch-sql_2025} attempt to correct these errors through a regeneration-based refinement process, formally defined as follows. 

\begin{definition}[Regeneration-based Refinement]
\label{def:regeneration-refinement}
Let $q_i \in Q$ denote the query at iteration $i (i \geq 0)$. 
The refinement process \[f_{\text{refine}}^{\text{regenerate}} : \mathcal{Q} \times D \to \mathcal{Q} ,\text{\;\;\;\;consists of: }\] 
\begin{itemize}[labelindent=0em,labelsep=0.5em,leftmargin=*]

    \item \textbf{Execution.}  
    The DBMS executes the current query \(q_i\).  
    If execution succeeds (i.e., \(D \vdash q_i \Rightarrow \llbracket q_i \rrbracket\)), it returns the result \(\llbracket q_i \rrbracket\). 
    Otherwise, if execution failed, it returns an execution error \(\epsilon(q_i)\).

    \item \textbf{Regeneration.}  
    The refiner updates the prompt $p_i$ into the new one $p_{i+1}$ with the error information and regenerates a new query:   
    \[p_{i+1} = (p_i, q_i, \epsilon(q_i)), \quad
    q_{i+1} = f_{\text{generate}}(p_{i+1})\] 
    This repeats until a valid query \(q_k\) satisfies  
    \(D \vdash q_k \Rightarrow \llbracket q_k \rrbracket.\)
\hfill $\Box$
\end{itemize}
\end{definition}

Figure~\ref{fig:regeneration-refinement} illustrates this common prompt structure adopted by most regeneration-based methods. 
However, as shown in Figure~\ref{fig:example-errors}(a), identical error messages \textit{Error: unknown attribute} may correspond to different underlying causes, each requiring a distinct refinement strategy. 
Moreover, as shown in Figure~\ref{fig:example-errors}(b), DBMS error messages often provide little or no actionable guidance, making it difficult for the model to determine the correct fix. 
To alleviate these limitations, several follow-up methods~\cite{ren_power_2025, talaei_chess_2024, xie_opensearch-sql_2025} enrich the error context.
RED-SQL~\cite{ren_power_2025} augments the error message with \textit{constraints violations} details to provide a more informative signal.
Meanwhile, CHESS-SQL~\cite{talaei_chess_2024} and
OpenSearch-SQL~\cite{xie_opensearch-sql_2025}
incorporate \textit{few-shot refinement examples} to guide the model toward more plausible refinements.

Despite these improvements, all regeneration-based methods fundamentally rely on the LLM to regenerate an error-free query.
Since the model offers no guaranty that a newly generated query will execute successfully,
the process often requires multiple iterations before convergence.
This iterative regeneration loop is both \textit{computationally expensive} and \textit{time-consuming}, significantly increasing LLM token usage and system latency.

\setlength{\FrameSep}{5pt}
\begin{center}
\vspace{-1mm}
\begin{minipage}{0.9\linewidth}
\begin{framed}
    {\sffamily\fontsize{8pt}{1pt}\selectfont
    \setlength{\baselineskip}{8pt}
    \noindent Refine the SQL based on the followings. \\[1.5pt]
    \noindent [Question] \\[-0.5pt]
    \noindent Show me the posts from people who are older than 24 \\[1.5pt]
    \noindent [Database Schema] \\[-0.5pt]
    \noindent \texttt{\textbf{CREATE TABLE} Users (ID int, name varchar, age int)} \\[-0.5pt]
    \noindent \texttt{\textbf{CREATE TABLE} Posts (UserID int, txt text, time date)} \\[1.5pt]
    \noindent [Generated SQL] \\[-0.5pt]
    \noindent \texttt{\textbf{SELECT} txt, time \textbf{FROM} Users \textbf{WHERE} age >= 24} \\[1.5pt]
    \noindent [Execution result]  \textit{Error: unknown attribute - txt} \\[1.5pt]
    \noindent [Few shot] \{\textit{refinement examples}\}
    }
\end{framed}
\vspace{-3mm}
\captionof{figure}{Example of a regeneration prompt.}
\label{fig:regeneration-refinement}
\end{minipage}
\end{center}
\vspace{-3mm}


%% file: sec_3.tex
\section{Safe Query Space}
\label{sec:search_space}

This section formalizes the search space in which SafeQL performs query refinement.
Unlike regenerating entire queries after each failure,
SafeQL systematically explores a structured space of refinements—syntactic transformations 
that incrementally repair an erroneous query based on its observed error.
SafeQL focuses on \textit{error-directed transformations}, i.e., targeted adjustments that resolve the specific cause of failure while retaining the original query’s structure and intent.
We first define by defining the atomic \textit{refinement step}, 
then construct the \textit{safe refinement tree} that organizes these steps into 
error-resolving paths, and finally formalize the \textit{safe query space} comprising executable, error-free queries.

\vspace{-3mm}
\subsection{Refinement Step}
\label{sec:search_paradigm:refinement_step}

We begin with a simplified SQL grammar that captures the core syntactic components.
Section~\ref{sec:system:grammar} later extends this grammar to support subquery, aliasing, and other complex SQL structures.

\setlength{\FrameSep}{1pt}
\begin{center}
\vspace{-1mm}
\begin{minipage}{0.8\linewidth}
\begin{framed}
    \[
    \begin{array}{@{}l@{\;\;}l@{}}
    \texttt{Query} & ::= \quad \texttt{SELECT S FROM F WHERE W} \\
    \texttt{F}      & ::= \quad \texttt{R | F JOIN R ON W} \\
    \texttt{W}      & ::= \quad \texttt{$a_1 \oplus^{\dagger} a_2$ | $a \oplus v$} \\
    \texttt{S}      & ::= \quad \texttt{$v$ | $a$ | f($a_1, a_2, ...$)}
    \end{array}
    \]
\end{framed}
\vspace{-3mm}
\captionof{figure}{Simplified SQL syntax.}
\vspace{-1mm}
\small{$^{\dagger}$~The symbol $\oplus$ denotes a comparison operator (e.g., =, <, >).}
\label{fig:syntax}
\vspace{-3mm}
\end{minipage}
\end{center}

Based on this grammar, we define the refinement stage a sequence of atomic refinement steps, where each step is a transformation between two SQL queries that modifies exactly one syntactic component while preserving the structural validity of the query.
Each refinement step corresponds to one of several canonical operations, formally described in Definition~\ref{def:refinement_step}.

\begin{definition}[Refinement Step, $\tightarrow{r}$]
A refinement step $r$ is the application of one of the following operations to a given SQL query.

\begin{itemize}[labelindent=0em,labelsep=0.3em,leftmargin=*,itemsep=0.2em]

    \item \textbf{Relation Refinement:}
    replace a relation $R$ with another $R'$:
    {\setlength{\abovedisplayskip}{1pt}
     \setlength{\belowdisplayskip}{1pt}
    \[
    \texttt{FROM}\; R \;\tightarrow{r}\; \texttt{FROM}\; R'
    \]

    \item \textbf{Join Refinement:}
    add a join with another relation $R'$ and a suitable condition $W'$:
    \[
    \texttt{FROM}\; R \;\tightarrow{r}\; \texttt{FROM}\; R \;\texttt{JOIN}\; R' \;\texttt{ON}\; W'
    \]

    \item \textbf{Attribute Refinement:}
    replace an attribute $a$ with another $a'$, in either \texttt{SELECT} or \texttt{WHERE}.  
    Typical cases include:
    \[
    \begin{array}{c}
        \texttt{SELECT}\; a \;\tightarrow{r}\; \texttt{SELECT}\; a' \\[4pt]
        \texttt{SELECT}\; f(a_1, a_2, \ldots) \;\tightarrow{r}\; \texttt{SELECT}\; f(a_1', a_2, \ldots)
        \\[4pt]
        \texttt{WHERE}\; a_1 \,\oplus\, a_2 \;\tightarrow{r}\; \texttt{WHERE}\; a_1' \,\oplus\, a_2 
        \quad\text{or}\quad
        \texttt{WHERE}\; a_1 \,\oplus\, a_2' \\[4pt]
        \texttt{WHERE}\; a \,\oplus\, v \;\tightarrow{r}\; \texttt{WHERE}\; a' \,\oplus\, v 
    \end{array}
    \]

    \item \textbf{Value Refinement:}
    replace a constant value $v$ with another $v'$:
    \[
    \texttt{WHERE}\; a \,\oplus\, v \;\tightarrow{r}\; \texttt{WHERE}\; a \,\oplus\, v'
    \]

    \item \textbf{Function Refinement:}
    replace a function symbol $f$ with another $f'$ while preserving its argument structure:
    \[
    \texttt{SELECT}\; f(a_1, a_2, \ldots) \;\tightarrow{r}\; \texttt{SELECT}\; f'(a_1, a_2, \ldots)
    \tag*{\text{\(\Box\)}}
    \]}
\end{itemize}
\label{def:refinement_step}
\end{definition}
\vspace{-5mm}

\subsection{Safe Refinement Tree}
\label{sec:search_paradigm:safe_refinement_tree}

Although each refinement step defines a syntactically valid transformation, 
executing all possible steps would lead to an excessively large and semantically irrelevant search space.
In practice, most execution errors arise from specific syntactic components of the query 
(e.g., relations, attributes, or functions), 
and only refinements targeting those components can effectively resolve the error~\cite{ning2024insights, chen2023text, shen2025study}. 
For instance, an \textit{unknown relation} error is typically resolved by replacing the relation in the \texttt{FROM} clause, whereas an \textit{unknown attribute} error is commonly addressed by replacing the invalid attribute in the \texttt{SELECT} clause.
These refinements differs from traditional rule-based generation~\cite{li_nalir_2014, yaghmazadeh_sqlizer_2017, saha_athena_2016}, whose handcrafted rules restrict what queries can be produced in the first place; SafeQL’s refinement act only after a query is generated, correcting errors without limiting the overall expressiveness of the query.
This observation motivates the notion of a \emph{safe refinement tree},
which defines the entire search space of SafeQL that can be explored through admissible refinement steps for a given error.

\begin{definition}[Safe Refinement Tree, $\mathcal{T}_{\text{safe}, q}$]
\label{def:safe_refinement_tree}
Given an query \(q\),
the safe refinement tree \(\mathcal{T}_{\text{safe}, q}\) 
is a directed tree that represents error-resolving refinement paths from $q$ and is defined as follows:

\begin{itemize}[labelindent=0em,labelsep=0.3em,leftmargin=*,itemsep=0em]
    \item The root of $\mathcal{T}_{\text{safe}, q}$ is the initial query $q$.
    \item An edge is added if the parent query produces an execution error.
    \item Each edge corresponds to a refinement step
    {\setlength{\abovedisplayskip}{2pt}
     \setlength{\belowdisplayskip}{2pt}
    \[
    q \tightarrow{r} q' \quad \text{for some} \quad r \in \mathcal{R}(\epsilon(q)).
    \]
    }
\end{itemize}
Here, $\mathcal{R}(\epsilon)$ denotes the set of admissible refinement operations 
allowed for each error $\epsilon$:

\[
\mathcal{R}(\epsilon) =
\left\{
\begin{array}{@{}l@{\quad}l@{}}
\left.
\begin{array}{@{}l@{}}
\makebox[9em][r]{\text{Relation Refinement}\;}
\end{array}
\right.
& \text{if } \epsilon = \textit{unknown relation} \\[5pt]
\vspace{1pt}
\left.
\begin{array}{@{}l@{}}
\makebox[9em][r]{\text{Relation Refinement}\;}\\
\makebox[9em][r]{\text{Join Refinement}\;}\\
\makebox[9em][r]{\text{Attribute Refinement}\;}
\end{array}
\right)
& \text{if } \epsilon = \textit{unknown attribute} \\[5pt]
\vspace{1pt}
\left.
\begin{array}{@{}l@{}}
\makebox[9em][r]{\text{Function Refinement}\;}\\
\makebox[9em][r]{\text{Attribute Refinement}\;}
\end{array}
\right)
& \text{if } \epsilon = \textit{unknown function} \\[5pt]
\vspace{1pt}
\left.
\begin{array}{@{}l@{}}
\makebox[9em][r]{\text{Value Refinement}\;}
\end{array}
\right.
& \text{if } \epsilon = \textit{empty result}  
\end{array}
\right.
\]
\begin{flushright}
\vspace{-1.7em}
$\Box$
\end{flushright}
\end{definition}

For \textit{unknown attribute} and \textit{unknown function} errors, 
multiple refinement operations are admissible because their causes are often indirect. 
An \textit{unknown attribute} may not only indicate a misspelled or mismatched column name, 
but also reveal that the necessary relation has not been joined at all. 
Likewise, an \textit{unknown function} may arise either from an undefined function name 
or from invalid arguments within its call---for example, when a missing attribute renders the function signature invalid.

Figure~\ref{fig:safe_refinement_tree} illustrates an example Safe Refinement Tree. 
The root query first fails with an \textit{unknown attribute} error because the attribute \texttt{txt} does not exist in the \texttt{Users} relation in the \texttt{FROM} clause. 
Accordingly, SafeQL explores only the admissible refinements for this error type—Attribute, Join, and Relation refinement—and Figure~\ref{fig:safe_refinement_tree} highlights three representative refinement steps:
\begin{description}[labelindent=0em,labelsep=0.3em,leftmargin=*,itemsep=0.2em]
\item[(1)] Relation Refinement (\texttt{Users} $\rightarrow$ \texttt{Posts});
\item[(2)] Join Refinement (\texttt{Users JOIN Posts ON UserID = ID});
\item[(3)] Attribute Refinement (\texttt{txt} $\rightarrow$ \texttt{name}).
\end{description}

After the attribute error is resolved—e.g., \texttt{txt} is resolved through the middle branch, where a Join Refinement introduces the missing \texttt{Posts} relation—the next level of the refinement tree produces a new failure in the function call.
Although \texttt{sum} is a valid aggregate function, it is not defined over string-typed arguments such as \texttt{txt}, leading to an \textit{unknown function} error.
At this refinement level, both function-level and attribute-level corrections are admissible, and two representative refinements are shown:
\begin{description}[labelindent=0em,labelsep=0.3em,leftmargin=*,itemsep=0.2em]
\item[(1)] Function Refinement (\texttt{sum} $\rightarrow$ \texttt{count});
\item[(2)] Attribute Refinement (\texttt{txt} $\rightarrow$ \texttt{txt\_len}).
\end{description}

\begin{figure}[h]
    \vspace{-3mm}
    \centerline{\includegraphics[width=3.5in]{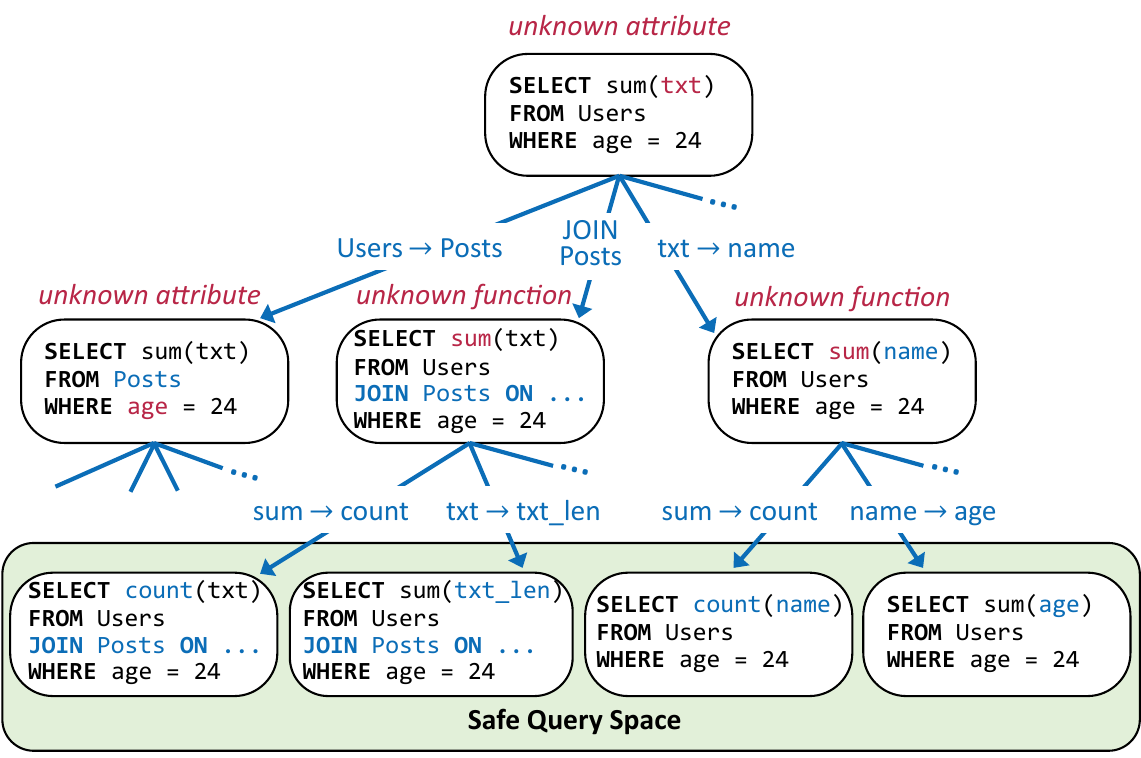}}
    \vspace{-3mm}
    \caption{Example of a Safe Refinement Tree.}
    \label{fig:safe_refinement_tree}
    \vspace{-3mm}
\end{figure}

The example above provides the intuition that each execution error can be
resolved by one of the admissible refinements defined for its error
type. To complete this intuition, we now formalize that such error-resolving
refinements are not merely illustrative but are guaranteed to exist within the
Safe Refinement Tree.

\begin{lemma}[Reachability of Error-resolving Refinements]
\label{lemma:error_refinement_reachability}
Given a database $D = (Rels, Atts, Vals, Typs, \Lambda, \Gamma)$ and the safe
refinement tree $\mathcal{T}_{\text{safe}, q}$, every execution error
$\epsilon(q)$ at a node $q \in \mathcal{T}_{\text{safe}, q}$ has a
corresponding refinement edge $(q, q')$ such that the refinement operation
$r \in \mathcal{R}(\epsilon(q))$ resolves the error under the database~$D$.
\end{lemma}

To formalize this guarantee, we introduce a lightweight operational semantics~\cite{hong2003introduction}
of SQL execution, in which an error judgment
$D \vdash q \Rightarrow \epsilon(q)$ is produced when the
corresponding error condition holds.  Formally:
\[
\frac{\textit{the error condition indicating a violation of execution premise}}
     {D \vdash q \Rightarrow \epsilon(q)}
\]
\vspace{-5mm}

\begin{proof}
We analyze each error type by making its operational error condition explicit and
showing that at least one refinement in $\mathcal{R}(\epsilon)$ always restores the
violated execution premise.  Let $\textit{relations}(F)$ denote the set of base
relations that appear in the \texttt{FROM} clause $F$.

\noindent\textbf{(1) \textit{Unknown Relation}.}
SQL execution requires that every relation referenced in $F$ exist in the
database schema:
\[
\frac{
    \exists R \in \textit{relations}(F),\; R \notin Rels
}{
    D \vdash \texttt{SELECT S FROM F WHERE W}
    \Rightarrow \textit{unknown relation}
}
\]
\vspace{-3mm}

A \textbf{Relation Refinement} replaces the invalid $R$ with a valid
relation $R' \in Rels$, thereby restoring the failed premise.
Thus, at least one refinement edge $(q,q')$ always exists.

\noindent\textbf{(2) \textit{Unknown Attribute}.}
An attribute must belong to the union of attribute
signatures of all relations in \textit{F}:
\[
\frac{
    a \notin \bigcup_{R \in \textit{relations}(F)} \Lambda(R)
}{
    D \vdash \texttt{SELECT a FROM F WHERE W}
    \Rightarrow \textit{unknown attribute}
}
\]
\vspace{-3mm}

This error can be repaired by one of the following refinements:
\begin{description}[labelindent=1em,labelsep=0.4em,leftmargin=*,itemsep=0.0em]

    \item[(2.1) Relation Refinement.]
    Replace a relation $R$ in the \texttt{FROM} clause with a relation
    $R'$ such that $a \in \Lambda(R')$, ensuring that the referenced
    attribute exists in the revised schema context.

    \item[(2.2) Join Refinement.]
    Add a relation $R'$ to the query such that $a \in \Lambda(R')$, so that
    $\textit{relations}(F)$ is extended with $R'$ and the attribute reference to be valid.

    \item[(2.3) Attribute Refinement.]
    Replace the unresolved attribute $a$ with an attribute $a'$ satisfying
    $a' \in \bigcup_{R \in \textit{relations}(F)} \Lambda(R)$, making the
    attribute reference to match one of the attributes provided by the
    current \texttt{FROM} clause.

\end{description}

In either case, the violated condition is recoverable:
if $a$ exists elsewhere in the schema, Relation or Join refinement brings the
corresponding relation into the query; otherwise, Attribute refinement replaces
$a$ with a valid attribute.  Hence at least one refinement edge $(q,q')$
always resolves the error.

\noindent\textbf{(3) \textit{Unknown Function}.}
A function expression may fail for one of two reasons: either the function itself is unsupported by the DBMS, or the function exists, but its argument expressions do not satisfy the type requirements imposed by its signature. Here, $Funcs$ denotes the set of all function symbols supported by
the DBMS.
\vspace{-1mm}
\[
\begin{gathered}
\frac{
    f \notin Funcs
}{
    D \vdash f(a_1,\ldots,a_k)
    \Rightarrow \textit{unknown function}
}
\\[1pt]
\frac{
    f : \tau_1,\ldots,\tau_k \rightarrow \tau_{\text{out}},\;
    \Gamma(a_i)=\tau'_i,\;
    (\tau'_1,\ldots,\tau'_k) \neq (\tau_1,\ldots,\tau_k)
}{
    D \vdash f(a_1,\ldots,a_k)
    \Rightarrow \textit{unknown function}
}
\end{gathered}
\]

This error can be repaired by one of the following refinements:
\begin{description}[labelindent=1em,labelsep=0.4em,leftmargin=*,itemsep=0.0em]

    \item[(3.1) Function Refinement.]
    If $f \notin Funcs$, replace $f$ with a supported function
    $f' \in Funcs$, restoring the violated condition.
    \item[(3.2) Attribute Refinement.]
    If the function symbol exists, but the argument types do not match its signature
$f : \tau_1,\ldots,\tau_k \rightarrow \tau_{\text{out}}$, we assume the schema
includes at least one attribute of each required type $\tau_i$.
Under this assumption, the violation is repaired by replacing one or more
arguments $a_i$ with attributes $a_i'$ such that $\Gamma(a_i') = \tau_i$,
thereby restoring the required types.

\end{description}

In either case, the violation is recoverable: it is replaced, and a type mismatch is fixed via argument refinement. Thus, at least one refinement edge $(q,q')$ resolves the error.

\noindent\textbf{(4) \textit{Empty Result}.}
This error occurs when the query is syntactically and semantically valid, but
its predicate evaluates to false for all tuples.  Typical causes include
overly restrictive comparison operators (e.g., $=$, $>$, $<$ applied in a way
that no tuple can satisfy), comparisons against values that do not exist in
the database, or contradictory filter conditions.  In such cases, the error
can be repaired by adjusting the value conditions, such as relaxing a
comparison, modifying a constant, or substituting a value that appears in the
data—so that the predicate becomes satisfiable.
Accordingly, a Value Refinement typically provides a refinement edge $(q,q')$
that adjusts the predicate to produce a non-empty result.

Taken together, these four cases show that every violated execution premise
can be repaired by an admissible refinement.
\qedhere
\end{proof}
\vspace{-2mm}

To illustrate how the operational semantics used in the proof relate to the 
refinement process, we present two examples that instantiate the refinement 
steps shown in Figure~\ref{fig:safe_refinement_tree}.

\noindent\textbf{Example 1 (Join Refinement; Case 2.2).}
{\setlength{\abovedisplayskip}{0pt}
 \setlength{\belowdisplayskip}{0pt}
\[
\frac{
    \texttt{txt} \notin \Lambda(\texttt{Users})
}{
    D \vdash \texttt{SELECT txt FROM Users}
    \Rightarrow \textit{unknown attribute}
}
\]
}
\vspace{-2mm}

If \(\texttt{txt} \in \Lambda(\texttt{Posts})\), then Join Refinement extends \texttt{FROM} clause:

\vspace{-3mm}
{\setlength{\abovedisplayskip}{0pt}
 \setlength{\belowdisplayskip}{0pt}
\[
\texttt{FROM Users}
\;\rightarrow\;
\texttt{FROM Users JOIN Posts ON UserID = ID}.
\]
}
The violated premise is restored because \(\texttt{txt} \in \Lambda(\texttt{Posts})\), 
This corresponds to the middle edge of the first level shown in Figure~\ref{fig:safe_refinement_tree}.

\noindent\textbf{Example 2 (Attribute Refinement; Case 3.2).}
\vspace{-2mm}

{\setlength{\abovedisplayskip}{0pt}
 \setlength{\belowdisplayskip}{0pt}
 \[
\frac{
    \texttt{sum} : \texttt{int} \rightarrow \texttt{int},\quad
    \Gamma(\texttt{name}) = \texttt{text},\quad
    \texttt{text} \neq \texttt{int}
}{
    D \vdash \texttt{SELECT sum(name) FROM Users}
    \Rightarrow \textit{unknown function}
}
\]
}
If the schema provides an attribute of the required type, such as 
\(\Gamma(\texttt{age}) = \texttt{int}\), then Attribute Refinement replaces 
the argument:
{\setlength{\abovedisplayskip}{1pt}
 \setlength{\belowdisplayskip}{0pt}
\[
\texttt{sum(name)} \;\rightarrow\; \texttt{sum(age)}.
\]
}
This corresponds to the rightmost edge of the second level shown in Figure~\ref{fig:safe_refinement_tree}.
Together, the Lemma~\ref{lemma:error_refinement_reachability} ensures that every erroneous node in the safe refinement tree
has at least one valid refinement leading to a safe (i.e., executable) query.
Consequently, the refinement process is guaranteed to reach 
a region of error-free queries, which we define as the \emph{safe query space}, formalized below.

\vspace{-1mm}
\begin{definition}[Safe Query Space]
Every leaf node $q'$ in the refinement tree $\mathcal{T}_{\text{safe}, q}$ satisfies 
$\llbracket q' \rrbracket \neq \epsilon(q)$. 
We refer to the set of all such leaf nodes as the \emph{Safe Query Space}:
\vspace{-0.5mm}
{\setlength{\abovedisplayskip}{0pt}
 \setlength{\belowdisplayskip}{0pt}
 \[
\mathcal{Q}_{\text{safe}} = 
\{\, q' \mid q' \text{ is a leaf node of } \mathcal{T}_{\text{safe}, q} 
\text{ and } \llbracket q' \rrbracket \neq \epsilon(q') \,\}.
\]
}
\end{definition}
\vspace{-5mm}

\begin{proof}
By construction of $\mathcal{T}_{\text{safe}, q}$, a query $q$ is refined when $\llbracket q \rrbracket = \epsilon(q)$.
Thus, if a node $q'$ has no children (i.e., it is a leaf), 
it must satisfy $\llbracket q' \rrbracket \neq \epsilon(q')$.
Therefore, all leaf nodes in $\mathcal{T}_{\text{safe}, q}$ are error-free, 
and their set constitutes the safe query space.
\end{proof}

In Figure~\ref{fig:safe_refinement_tree}, the green-highlighted queries correspond exactly 
to such leaf nodes in $\mathcal{T}_{\text{safe}, q}$ and thus serve as concrete instances 
of queries contained in the safe query space.

%% file: sec_4.tex
\section{Search Algorithm}
\label{sec:search_algorithm}

This section explains how SafeQL navigates the safe query space
to identify the most semantically faithful executable query.
We first define the search objective based on semantic distance,
then describe a best-first search algorithm that prioritizes
refinements closest to the original query.

\subsection{Search Objective under Semantics}
\label{sec:search_paradigm:search_objectives}

Having defined the \emph{safe query space} 
$\mathcal{Q}_{\text{safe}}$,
we now formalize the objective that governs the search process within this space.
Even among valid candidates in $\mathcal{Q}_{\text{safe}}$,
refinements may differ in how well they preserve the original intent.
To quantify such a deviation, we define a \textit{semantic distance}
that jointly measures structural change and embedding-level drift.

\begin{definition}[Semantic Distance, $\delta$]
For two queries $q_i$ and $q_j$, the semantic distance $\delta(q_i, q_j)$ is defined as
    {\setlength{\abovedisplayskip}{2pt}
     \setlength{\belowdisplayskip}{2pt}
    \[
    \delta(q_i, q_j)
    =
    \alpha \cdot d_{\text{struct}}(q_i, q_j)
    + (1 - \alpha) \cdot d_{\text{embed}}(q_i, q_j),
    \;\; \alpha \in [0,1].
    \]
    }
\vspace{-4mm}
\begin{itemize}[leftmargin=*]
    \item 
    $d_{\text{struct}}(q_i, q_j)$ measures structural deviation 
    (e.g., relation or join refinements) using the normalized tree edit distance~\cite{zhang_simple_1989}
    between the abstract syntax trees of $q_i$ and $q_j$, a standard metric for comparing hierarchical structures such as XML and JSON.
    
    \item 
    $d_{\text{embed}}(q_i, q_j)$ measures semantic deviation based on embeddings: 
    \vspace{-2mm}
    {\setlength{\abovedisplayskip}{0pt}
     \setlength{\belowdisplayskip}{0pt}
    \[
    d_{\text{embed}}(q_i, q_j)
    = \sum_{(t \rightarrow t')}
    \|v_t - v_{t'}\|,
    \]
    }
    where $(t \!\rightarrow\! t')$ denotes each replaced token 
    (relation, attribute, or value) and $v_t$ is the embedding vector of token $t$
    computed by a neural embedding model~\cite{reimers_sentence-bert_2019}.
\hfill $\Box$
\end{itemize}
\end{definition}

This hybrid formulation is conceptually analogous to the sparse-dense hybrid retrieval 
in information retrieval, which combines complementary lexical and semantic signals
to improve ranking robustness~\cite{santhanam_colbertv2_2022, formal_splade_2021}.
\textcolor{notice}{
Similarly, combining $d_{\text{struct}}$ and $d_{\text{embed}}$
provides a more robust and balanced measure of syntactic fidelity and semantic preservation, while remaining a heuristic rather than a strict notion of semantic equivalence.
The former enforces grammatical consistency,
while the latter maintains semantic proximity in embedding space.
}
The weighting parameter $\alpha$ controls this trade-off,
and we empirically validate its effect in Section~\ref{sec:evaluation:efficiency}.

With semantic distance formalized, we now explain how SafeQL leverages it during the refinement process.
Unlike regeneration-based methods that repeatedly apply $\mathcal{Q} \rightarrow \mathcal{Q}$ by regenerating full SQLs after every failure,
SafeQL redefines refinement as a guided search from the general query space $\mathcal{Q}$
into the executable subset $\mathcal{Q}_{\text{safe}}$ under the database $D$.
In this view, each erroneous query serves not as a terminal failure but as a starting point from which SafeQL identifies the nearest semantically consistent counterpart in $\mathcal{Q}_{\text{safe}}$ under the distance metric $\delta$.
We formalize this as follows:

\begin{definition}[Search-based Refinement]
\label{def:search-based-refinement}

The refinement process finds the safe query nearest to the input query $q$ within a  distance~$\delta$.
\[
f_{\text{refine}}^{\text{search}} :
\mathcal{Q} \times D \to \mathcal{Q}_{\text{safe}},
\quad
q_{\text{safe}} = \arg\min_{q' \in \mathcal{Q}_{\text{safe}}} \delta(q, q').
\]
\end{definition}

\subsection{Best-first Search-based Refinement}
\label{sec:search_algorithm:best_first_search}

The overall refinement process follows a best-first exploration of the safe query space,
as shown in Algorithm~\ref{alg:best_first_search}.
Unlike exhaustive search, which explore all refinements uniformly, best-first search naturally limits exploration in a branch-and-bound style~\cite{narendra1977branch}, accelerating convergence.
\textcolor{notice}{
The search initializes a priority queue ordered by semantic distance from the LLM-generated initial query, expanding the closest candidates first. 
This design takes the initial query as the reference point under the assumption that it reflects the LLM’s interpretation of the user’s intent and thus serves as a reasonable available signal for recovering the original intent. 
When this initial query is fundamentally misgenerated, search-based refinement alone is insufficient, and SafeQL therefore incorporates a hybrid refinement strategy to regenerate the query and establish a better starting point, as described later in Section~\ref{sec:system:hybrid}.}

During each iteration of the loop (Lines~\ref{alg:best_first_search:loopStart}-\ref{alg:best_first_search:loopEnd}),
the algorithm dequeues the nearest candidate query $q$ (Line~\ref{alg:best_first_search:pop}) and attempts to execute it
within the \texttt{tryExecute} block (Lines~\ref{alg:best_first_search:tryStart}).
If execution succeeds, the algorithm immediately returns $q$
as the nearest safe query (Line~\ref{alg:best_first_search:return}),
thus guaranteeing that the first valid query discovered is the one most semantically similar to the original input.

When execution fails, SafeQL captures the corresponding error type $\epsilon(q)$ within the \texttt{catchError} block
(Lines~\ref{alg:best_first_search:catchStart}-\ref{alg:best_first_search:catchEnd})
and uses it to guide subsequent refinements.
The \textsc{GenerateRefinements} procedure (Line~\ref{alg:best_first_search:generate})
selects the appropriate refinement operations (e.g., relation, join, or attribute refinements) based on $\epsilon$,
while \textsc{PruneRefinements} (Line~\ref{alg:best_first_search:pruneStart})
filters out redundant or semantically invalid refinements.
Each surviving refinement $r$ is then applied to produce a new candidate query $q'$ (Line~\ref{alg:best_first_search:applyStart}),
which is enqueued for future exploration only if it has not been visited previously
(Lines~\ref{alg:best_first_search:enqueueStart}).

\vspace{-2mm}

%% file: sec_5.tex
\section{Pruning Methods}
\label{sec:pruning_methods}

Although all queries in the safe space are executable, the space can grow rapidly as each refinement generates many candidates.
For each error, all schema-level elements (relations, attributes, values) are considered, causing a single step to branch into dozens of candidates.
Value refinements can expand even further, since every possible value in the database becomes a potential substitution.

For example, in Figure~\ref{fig:pruning-example}, an \textit{unknown attribute} error for \texttt{user} leads SafeQL to try all attributes (e.g., \texttt{name}, \texttt{ID}, \texttt{age}).
Even after correction, \texttt{WHERE name = `leo'} may produce an \textit{empty result}, triggering another round of refinements where \texttt{`leo'} is substituted with every possible value in the table—\texttt{`Leo'}, \texttt{`Hana'}, \texttt{`James'}, and so on.
This results in combinatorial growth of candidates across refinement layers.

\vspace{-2mm}
\begin{proposition}[Exponential Size of Safe Query Space]
For a safe refinement tree $\mathcal{T}_{\textnormal{safe}, q}$ with branching factor $r$ and depth $d$, the total size of the search space is $\mathcal{O}(r^d)$, since each refinement can generate $r$ new candidates at every level.
\end{proposition}
\vspace{-2mm}

This exponential growth makes exhaustive enumeration infeasible even for moderately complex databases.
To mitigate this issue, SafeQL introduces pruning mechanisms that restrict the number of refinements explored at each step.
These strategies—type-based pruning and top $K$ refinement pruning—jointly suppress unpromising branches in the refinement tree, effectively reducing fan-out while preserving meaningful candidates.

\vspace{-2mm}
\begin{algorithm}[t]
\SetAlgoLined
\SetArgSty{textnormal}
\KwIn{$q$ \tcc*[f]{original query}} 
\KwIn{$\alpha$ \tcc*[f]{structure–embedding weight factor}}
\KwIn{\textit{refineQueue}\!\!\tcc*[f]{priority queue ordered by $\delta(q, q')$}}

\SetKwFunction{Fexecute}{execute}
\SetKwFunction{FgenerateRef}{generateRefinements}
\SetKwFunction{FpruneRef}{pruneRefinements}
\SetKwFunction{FapplyRef}{applyRefinements}
\SetKwFunction{Fdistance}{calculateDistance}
\SetKwFunction{FsearchBasedRefine}{searchBasedRefine}

\SetKwBlock{Try}{tryExecute\,(\textnormal{$q_{\text{current}}$)}}{end}
\SetKwBlock{Catch}{catchError}{end}
\SetKwProg{Fn}{Function}{:}{end}

push ($q$, 0) into \textit{refineQueue}\; \label{alg:best_first_search:initStart}
mark $q$ as visited\; \label{alg:best_first_search:initEnd}

\Fn{\FsearchBasedRefine{$q$}}{
\While{\textit{refineQueue} is not empty}{ \label{alg:best_first_search:loopStart}
    $q_{\text{current}} \leftarrow \argminA_{q' \in \textit{refineQueue}} \delta(q, q')$\; \label{alg:best_first_search:pop}
    mark $q_{\text{current}}$ as visited\;
    \Try{ \label{alg:best_first_search:tryStart}
        \Return $q_{\text{current}}$; \tcc*[f]{nearest safe query found} \label{alg:best_first_search:return}
    }\Catch{ \label{alg:best_first_search:catchStart}
        $\mathcal{R} \leftarrow$ \FgenerateRef{$\epsilon(q_{\textnormal{current}})$}\; \label{alg:best_first_search:generate}
        $\mathcal{R}_{\text{pruned}} \leftarrow$ \FpruneRef{$\mathcal{R}$}\; \label{alg:best_first_search:pruneStart}
        \ForAll{$r \in \mathcal{R}_{\text{pruned}}$}{
            $q' \leftarrow$ \FapplyRef$(q_{\text{current}}, r)$\; \label{alg:best_first_search:applyStart}
            \If{$q'$ not visited}{
                $\delta \leftarrow$ \Fdistance{$q, q', \alpha$}\;
                push $(q', \delta)$ into \textit{refineQueue}\; \label{alg:best_first_search:enqueueStart}
            }
        }
    } \label{alg:best_first_search:catchEnd}
} \label{alg:best_first_search:loopEnd}
}

\caption{Best-first Search-based Refinement}
\label{alg:best_first_search}
\end{algorithm}

\begin{figure}
\vspace{-3mm}
\centerline{\includegraphics[width=3.0in]{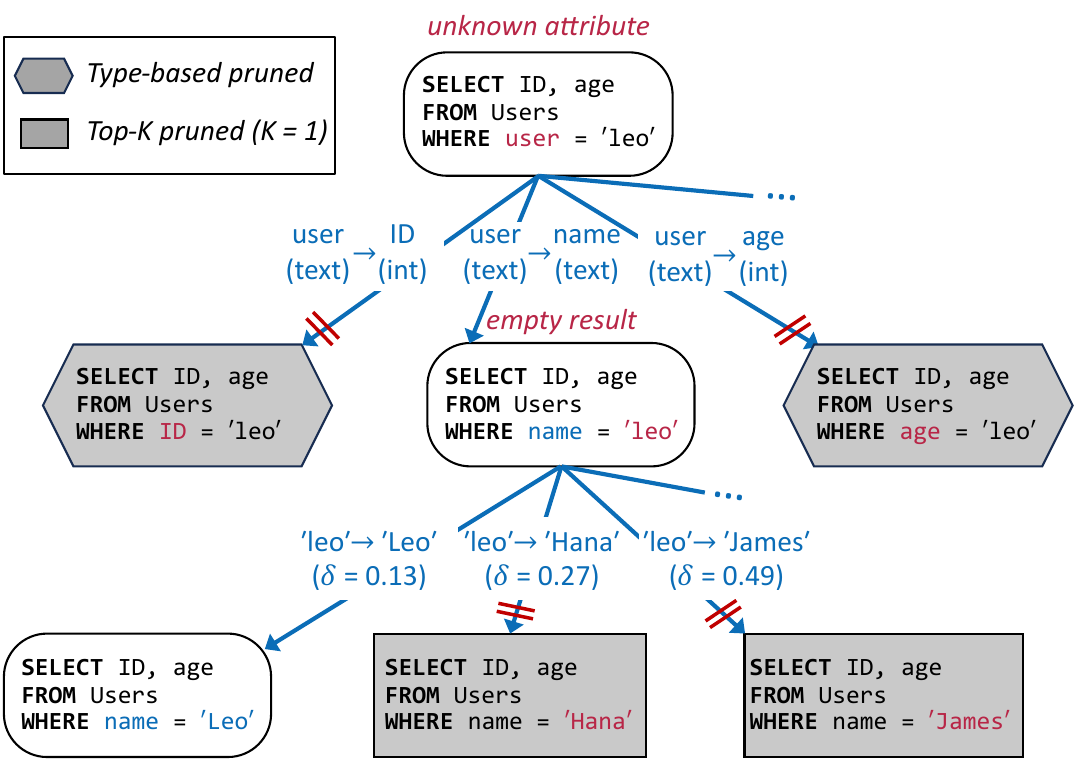}}
\vspace{-2mm}
\caption{Example of pruning in the safe refinement tree.}
\label{fig:pruning-example}
\vspace{-5mm}
\end{figure}

\subsection{Type-based Refinement Pruning}
\label{sec:search_algorithm:type_based_pruning}

Type information provides a strong pruning signal because any type mismatch immediately yields an invalid query.
In SQL, each operator and function has a fixed input signature;  
thus, if a substituted attribute or value violates this signature, execution inevitably results in an error.
SafeQL leverages the typing environment
$\Gamma$
to discard such candidates before semantic evaluation.

Algorithm~\ref{alg:type_based_pruning} presents this mechanism.
It validates refinements under $\Gamma$ and prunes type-inconsistent candidates.
In Figure~\ref{fig:pruning-example}, \texttt{WHERE user = `leo'} yields \texttt{name}, \texttt{ID}, and \texttt{age}, with
$\Gamma(\texttt{name})=\texttt{TEXT}$ and $\Gamma(\texttt{ID})=\Gamma(\texttt{age})=\texttt{INT}$.
Applying \textsc{TypeCheck}, only \texttt{name} is retained since \texttt{`leo'} is \texttt{TEXT}.
This prevents non-executable branches from entering the search.

\vspace{-2mm}
\subsection{Top-$K$ Refinement Pruning}
\label{sec:search_algorithm:top_k_pruning}

Even after type-based pruning, many refinements remain type-consistent but semantically distant.
To improve efficiency, SafeQL ranks candidates by embedding similarity and retains only the top-$K$ per category.
This is particularly beneficial for value refinements, where the out-degree is large and most database values are semantically unrelated to the query; in such cases, embedding-based similarity provides an effective signal for filtering out implausible candidates.
A smaller $K$ yields higher efficiency but risks over-pruning, while a larger $K$ improves accuracy at increased search cost, a tradeoff we examine empirically in Section~\ref{sec:evaluation:efficiency}.

Algorithm~\ref{alg:top_k_pruning} summarizes this mechanism.
SafeQL retrieves candidates, ranks them by semantic distance, and retains only the top-$K$.
Examples are shown in Figure~\ref{fig:pruning-example}.
This pruning suppresses unpromising branches and focuses refinement on relevant candidates.


{\setlength{\textfloatsep}{4pt}
 \setlength{\floatsep}{4pt}
 \setlength{\intextsep}{4pt}
 \setlength{\abovecaptionskip}{0pt}
 \setlength{\belowcaptionskip}{0pt}
\begin{algorithm}
\SetAlgoLined
\SetArgSty{textnormal}
\SetKwProg{Fn}{Function}{:}{end}
\SetKwFunction{FpruneRef}{pruneRefinements}
\SetKwFunction{FcheckType}{typeCheck}
\SetKw{Return}{return}

\Fn{\FpruneRef{$\mathcal{R}$}}{
    \ForEach{$r \in \mathcal{R}$}{
        \Switch{$r$}{
            \Case{Attribute Refinement: $a_1 \oplus a_2 \tightarrow{r} a_1' \oplus a_2$}{
                $\tau_1,\;\tau_2 \leftarrow \Gamma(a_1'),\; \Gamma(a_2)$\; \label{alg:type_based_pruning:fetch}
                \If{\FcheckType{$\oplus,\; \tau_1,\; \tau_2$} pass}{add $r$ to $\mathcal{R}_{\textnormal{pruned}}$\; \label{alg:type_based_pruning:check}}
            }
            \Case{Value Refinement: $a \oplus v \tightarrow{r} a \oplus v'$}{
                $\tau_a,\;\tau_v \leftarrow \Gamma(a),\; \Gamma(v')$\;
                \If{\FcheckType{$\tau_a,\; \tau_v$} pass}{add $r$ to $\mathcal{R}_{\textnormal{pruned}}$\;}
            }
            \Case{Function Refinement: $f \tightarrow{r} f'$}{
                \If{\FcheckType{$f',\; \Gamma(a_1),\; \ldots$} pass}{add $r$ to $\mathcal{R}_{\textnormal{pruned}}$\;}
            }}}
}
\caption{\textsc{Type-Based Refinement Pruning}}
\label{alg:type_based_pruning}
\end{algorithm}
}

{
\setlength{\textfloatsep}{4pt}
 \setlength{\floatsep}{4pt}
 \setlength{\intextsep}{4pt}
 \setlength{\abovecaptionskip}{0pt}
 \setlength{\belowcaptionskip}{0pt}
\vspace{-2mm}
\begin{algorithm}
\SetAlgoLined
\SetArgSty{textnormal}
\SetKwIF{If}{ElseIf}{Else}{if}{}{else if}{else}{}
\SetKwProg{Fn}{Function}{:}{end}
\SetKwFunction{FpruneRef}{pruneRefinements}
\SetKw{Return}{return}

\KwIn{$K$ \tcc*[f]{number of nearest neighbors to explore}}

\Fn{\FpruneRef{$\mathcal{R}$}}{
    \label{alg:top_k_pruning:start}  

    \ForEach{$r \in \mathcal{R}$}{
      \label{alg:top_k_pruning:loopStart}

        \Switch{$r$}{

            \Case{Relation Refinement: $R \tightarrow{r} R'$}{
                \If{
                  $|\{x \in Rels : \|v_R - v_x\| \le \|v_R - v_{R'}\|\}| \le K$
                }{
                    add $r$ to $\mathcal{R}_{\textnormal{pruned}}$\;
                }
            }

            \Case{Attribute Refinement: $a \tightarrow{r} a'$}{
                \If{
                  $|\{x \in Atts : \|v_a - v_x\| \le \|v_a - v_{a'}\|\}| \le K$
                }{
                    add $r$ to $\mathcal{R}_{\textnormal{pruned}}$\;
                }
            }

            \Case{Value Refinement: $v \tightarrow{r} v'$}{
                \If{
                  $|\{x \in Vals : \|v_v - v_x\| \le \|v_v - v_{v'}\|\}| \le K$
                }{
                    add $r$ to $\mathcal{R}_{\textnormal{pruned}}$\;
                }
            }

            \Case{Function Refinement: $f \tightarrow{r} f'$}{
                \If{
                  \!\!$|\{x \in Func : \|v_f - v_x\| \le \|v_f - v_{f'}\|\}| \le K$ 
                }{
                    add $r$ to $\mathcal{R}_{\textnormal{pruned}}$\;
                }
            }

        }

      \label{alg:top_k_pruning:loopEnd}
    }

}
\caption{\textsc{Top-$K$ Refinement Pruning}}
\label{alg:top_k_pruning}
\end{algorithm}
}

\vspace{-4mm}

%% file: sec_6.tex
\section{System Architecture Design}
\label{sec:system}

Although the core of SafeQL lies in its search-based refinement algorithm,
realizing it as a refinement framework for building practical text-to-SQL systems requires addressing several system-level challenges.  
Unlike a conceptual refinement algorithm, an actual implementation must integrate deeply with the DBMS, eliminate redundant computation, remain robust against fundamentally incorrect queries, support syntactically complex SQL structures, and stay compatible with diverse Text-to-SQL systems. Consequently, SafeQL’s architecture is built around five key design considerations:
\begin{itemize}[labelindent=0em,labelsep=0.5em,leftmargin=*,itemsep=0em]
\item \textbf{In-DBMS integration} for precise error detection.
\item \textbf{Systematic optimization} for scalable refinement.
\item \textbf{Hybrid refinement} that complements search with regeneration.
\item \textbf{Grammar extensions} enabling broad SQL coverage. 
\item \textbf{Compatibility} with diverse Text-to-SQL systems as a post-generation refinement layer. We describe each aspect below.
\end{itemize}

\vspace{-3mm}
\subsection{In-DBMS Integration}
\label{sec:system:integration}

\textcolor{notice}
{
External Text-to-SQL systems rely on error strings, which may mention identifiers (e.g., unknown relation \texttt{R}) but do not reliably indicate error locations; some DBMSs provide no positions, and even when they do (e.g., position 15), they refer only to character offsets and cannot distinguish multiple occurrences such as \texttt{SELECT R.a FROM R}.
Locating the correct occurrence thus requires additional parsing, which is ambiguous under aliasing, nested queries, or repeated references.
SafeQL instead performs refinement inside the DBMS over the \textit{Abstract Syntax Tree (AST)}.
When a query $q_i$ fails, the analyzer directly identifies the corresponding AST node.
Modifications are restricted to these nodes; since the AST encodes attribute–relation links (e.g., \texttt{R.a} refers to \texttt{R}), updates can be applied without affecting unrelated occurrences, and each refined query is re-analyzed and validated.
This enables precise and safer structurally grounded refinement beyond text-based approaches, despite slightly higher implementation complexity.
}

\vspace{-3mm}
\subsection{Systematic Optimization}
\label{sec:system:optimization}

Executing refinement inside the DBMS introduces a new bottleneck: repeated embedding and similarity computations across overlapping query fragments.
To eliminate such redundancy, SafeQL incorporates two in-database optimization layers:

\begin{itemize}[labelindent=0em,labelsep=0.5em,leftmargin=*,itemsep=0em]
    \item \textbf{Embedding Cache.}
    Embeddings of relations, attributes, and values are precomputed and stored on disk, while frequently accessed items are cached in memory. This prevents repeated embedding computation across refinements.
    
    \item \textbf{Vector Index.}
    To efficiently perform nearest-neighbor search (particularly for Top-$K$ pruning),
    SafeQL builds an HNSW-based vector similarity index over cached embeddings.
\end{itemize}

\textcolor{notice}
{
Since these structures are maintained in the DBMS, tuple and embedding updates are consistently reflected in the vector index via a vector extended relational system.
Optimizing update maintenance via incremental techniques~\cite{xu2023spfresh, Yu26Greater} is left for future work.
}

\vspace{-6mm}
\subsection{Hybrid Refinement Strategy}
\label{sec:system:hybrid}

\textcolor{notice}
{
Even with precise in-database refinement, some queries are fundamentally misgenerated.
To handle these, SafeQL adopts a hybrid strategy combining search-based and regeneration-based refinement.
When the number of refinement steps reaches a threshold (default 100) without producing a valid query, SafeQL triggers a regeneration loop.
Such cases typically indicate that the initial query is too far from a valid structure for localized refinement to succeed, in which case an LLM generates a new candidate conditioned on the original input, the current query, and the DBMS error feedback.
This also addresses LLM-side errors, such as invalid or even non-parsable generations by reinitializing the search from a new candidate.
The relative contribution of search-based refinement and hybrid fallback is analyzed in Table~\ref{tab:safeql_delta} and Subsection~\ref{sec:evaluation:main results:bird}.
This hybrid design ensures robustness to deeply misgenerated queries without undermining the efficiency of search-based refinement.
}

\vspace{-1mm}
\subsection{Grammar Extensions}
\label{sec:system:grammar}

Practical Text-to-SQL systems must handle more complex SQL than the simplified grammar presented in Section~\ref{sec:search_space}.
\textcolor{notice}
{
SafeQL therefore extends its support to standard SQL constructs appearing in benchmarks such as Bird~\cite{li2024can} and Spider~\cite{yu-etal-2018-spider}, including:
\begin{itemize}[labelindent=0em,labelsep=0.5em,leftmargin=*,itemsep=0em]
    \item \textbf{Nested subqueries:} 
    supports recursively scoped queries by refining subqueries in an inside-out manner. 
A subquery $q_s$ is recursively refined into $q_s'$ (Line~\ref{alg:nested:recurse}) and replaced (Line~\ref{alg:nested:replace}), after which refinement continues on the outer query (Line~\ref{alg:nested:return}). 
The output $\Lambda(q_s')$ is exposed to the DBMS analyzer (Line~\ref{alg:nested:schema}), enabling proper scope resolution for outer references.
    \item \textbf{Alias refinement:} 
    extends relation and attribute refinement to handle alias normalization (e.g., \texttt{R.a}).
    \item \textbf{Clause-level extensions:} 
    support clauses such as \texttt{GROUP BY}, \texttt{ORDER BY}, and \texttt{LIMIT}. 
    For example, \texttt{ORDER BY O} takes the form ($\texttt{O} ::= a \;|\; f(a_1, a_2, \ldots)$) and is refined as: \\ 
    $\texttt{ORDER BY}\; a \tightarrow{r} a',\;
f(a_1, ...) \tightarrow{r} f(a_1', ...),\;
f(a_1, ...) \tightarrow{r} f'(a_1, ...)$.
\texttt{GROUP BY} and \texttt{LIMIT} are handled analogously.
\end{itemize}
}

\vspace{2mm}

{
{\setlength{\textfloatsep}{4pt}
 \setlength{\floatsep}{4pt}
 \setlength{\intextsep}{4pt}
 \setlength{\abovecaptionskip}{0pt}
 \setlength{\belowcaptionskip}{0pt}
\SetAlCapFnt{\color{notice}}
\SetAlCapNameFnt{\color{notice}}
\vspace{-2mm}
\begin{algorithm}

\SetAlgoLined
\SetArgSty{textnormal}
\SetKwProg{Fn}{Function}{:}{end}
\SetKwFunction{FrecursiveRefine}{recursiveRefine}
\SetKwFunction{FsearchBasedRefine}{searchBasedRefine}
\SetKw{Return}{return}
\Fn{\FrecursiveRefine{$q$}}{
    \If{$q$ contains a subquery $q_s$}{ \label{alg:nested:if}
        $q_s' \leftarrow$ \FrecursiveRefine{$q_s$}\; \label{alg:nested:recurse}

        $\Lambda(q_s') \leftarrow \{a_1, a_2, \dots \}$; \label{alg:nested:schema}
        \tcc*[f]{$q_s' = \texttt{SELECT } a_1, a_2, \dots$}

        replace $q_s$ in $q$ with $q_s'$\; \label{alg:nested:replace}
    }
    \Return \FsearchBasedRefine{$q$}; \label{alg:nested:return}
}
\caption{\textcolor{notice}{\textsc{Nested Subquery Refinement}}}
\label{alg:nested_subquery_refinement}
\end{algorithm}
\vspace{-4mm}
}

\subsection{Compatibility with Text-to-SQL Systems}
\label{sec:system:compatibility}

SafeQL integrates with diverse Text-to-SQL systems as a post-generation refinement layer.
It attaches to prompt-based systems without modifying prompts and applies to agent-based systems by refining their final queries to ensure safe execution in the DBMS.


\vspace{-2mm}
\subsection{Architecture}
\label{sec:system:architecture}

Figure~\ref{fig:architecture} summarizes how these components interact within the SafeQL architecture.  
When a query \( q_i \) fails over the database \( D \), the query refiner operates inside the DBMS to locate the error source \( \epsilon(q_i) \) and generate refined candidates \( q_i' \) that correct the erroneous AST nodes (Section~\ref{sec:system:integration}), while leveraging cached embeddings and vector indices for efficient similarity computation (Section~\ref{sec:system:optimization}).  

\vspace{-2mm}
\begin{figure}[hbtp]
    \centerline{\includegraphics[width=2.8in]{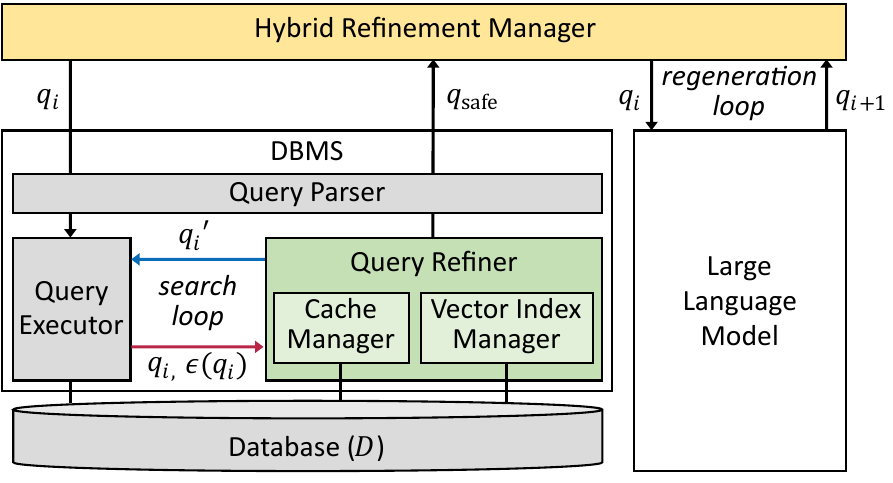}}
    \vspace{-3mm}
    \caption{
        Architecture of the SafeQL framework.
    }
    \label{fig:architecture}
\vspace{-4mm}
\end{figure}

Within the \textit{search loop}, the query refiner iteratively explores candidate refinements until a valid query \( q_{\text{safe}} \) is found. 
If unsuccessful until a threshold is reached, control moves to the \textit{regeneration loop}, where an LLM regenerates a new query \( q_{i+1} \) guided by the hybrid refinement manager (Section~\ref{sec:system:hybrid}).  
Throughout the process, all intermediate queries are parsed and validated under the extended SQL grammar (Section~\ref{sec:system:grammar}), ensuring that even complex transformations remain syntactically correct and executable.  
Together, these mechanisms elevate SafeQL from an algorithmic concept to a practical, database-native Text-to-SQL refinement framework.

\noindent
\textbf{Implementation: }
SafeQL is implemented on top of PostgreSQL’s RawStmt structure, the canonical AST representation of SQL queries.  
It is packaged as a PostgreSQL~\cite{postgresql} extension, using fastembed~\cite{fastembed} with the bge-base-en-v1.5 model~\cite{bge-base-en-v1.5} for embedding computation and pgvecto.rs~\cite{pgvecto.rs} for efficient vector indexing inside the DBMS.


%% file: sec_7.tex
\section{Experimental evaluation}
\label{sec:evaluation}

In this section, we evaluate SafeQL in terms of two key dimensions: (1) accuracy, measured by execution accuracy, and (2) efficiency, measured by elapsed time and the number of LLM tokens consumed.
Our goal is to demonstrate that SafeQL improves accuracy while reducing expensive LLM invocations across various settings.

\subsection{Experimental Setup}

\subsubsection{Benchmarks.} 
\begin{description}[labelindent=0em,labelsep=0.5em,leftmargin=*,itemsep=0em]
\vspace{-2.5pt}
\item[\textbf{Bird}]
\textcolor{notice}
{
\!\!\!~\cite{li2024can} is a widely used cross-domain Text-to-SQL benchmark and our primary evaluation dataset.
To align with our PostgreSQL-based implementation, we migrated the Bird to PostgreSQL. 
We evaluate both the full dev set (1,534 queries) and the mini dev set (500 queries), for main and detailed evaluation, respectively.
Both dev sets span 11 databases and cover diverse queries, including joins, grouping, and subqueries.
}

\item[\textbf{Spider}]
\!\!\!~\cite{yu-etal-2018-spider} is another major cross-domain benchmark, consisting of 10,181 text-to-SQL pairs across 200 relational databases spanning more than 130 domains. 
For our experiments, we migrated Spider to PostgreSQL to verify SafeQL’s generalizability across different benchmarks and database environments.
\end{description}

\subsubsection{Testbed systems.} 
We integrate SafeQL with two representative Text-to-SQL approaches: prompt-based and agent-based methods. Specifically, we adopt DAIL-SQL and OpenSearch-SQL as the state-of-the-art systems for these approaches, respectively.

\begin{description}[labelindent=0em,labelsep=0.5em,leftmargin=*,itemsep=0em,itemsep=0em]

\item[\textbf{DAIL-SQL}]
\!\!\!~\cite{gao_text--sql_2024} is a prompt-based system that optimizes question representation, example selection, and example organization for in-context learning.
As the state-of-the-art prompt-based method on the Spider benchmark, it serves as a strong baseline for evaluating SafeQL’s refinement in prompting-based settings.

\item[\textbf{OpenSearch-SQL}]
\!\!\!~\cite{xie_opensearch-sql_2025} is an agent-based system that decomposes the Text-to-SQL process into four modules—\textit{Preprocessing, Extraction, Generation, and Refinement}—linked through a consistency alignment mechanism.
It leverages dynamic few-shot learning with self-taught Query-CoT-SQL pairs and structured few-shot templates, achieving state-of-the-art performance on the Bird benchmark.
We integrate SafeQL into OpenSearch-SQL to evaluate its refinement behavior under multi-agent coordination.

\end{description}

\subsubsection{Comparison methods.} 
To isolate the effect of refinement, we compare SafeQL with the refinement methods of five prominent Text-to-SQL systems---DIN-SQL, MAC-SQL, RED-SQL, CHESS-SQL, and OpenSearch-SQL---all of which adopt regeneration-based paradigms.

\begin{itemize}[labelindent=0em,labelsep=0.5em,leftmargin=*,itemsep=0em,itemsep=0em]

\item \textbf{DIN-SQL}
\!~\cite{pourreza_din-sql_2023} refines SQL without execution, assuming potential errors.
Among its two prompt variants, GENERIC and GENTLE, we use GENTLE, which performed better in our experiments.

\item \textbf{MAC-SQL}
\!~\cite{wang2025mac} regenerates SQL using DBMS error messages.

\item \textbf{RED-SQL}
\!~\cite{ren_power_2025} refines SQL using constraint-based violation reports rather than raw DBMS errors,
providing data-aware feedback that helps the LLM correct semantic inconsistencies.

\item \textbf{CHESS-SQL} 
\!~\cite{talaei_chess_2024} uses few-shot refinement examples combined with an execution error message.

\item \textbf{OpenSearch-SQL}
\!~\cite{xie_opensearch-sql_2025} also uses few-shot examples, but provides more structured templates tailored to specific error messages.

\end{itemize}

\vspace{-3mm}
\setlength{\tabcolsep}{1.8pt}
\renewcommand{\arraystretch}{1.0}
\begin{table}[H]
\caption{Comparison of refinement methods.}
\vspace{-4mm}
\label{tab:comparison_methods}
\centering
\begin{tabular}{l|c|c|c}
\hline
\textbf{\makecell[l]{Refinement \\ Methods}} & 
\textbf{\makecell{Refinement \\ Paradigm}} & 
\textbf{\makecell{Use of \\ Error Msg}} & 
\textbf{\makecell{Use of \\ Examples}} \\
\hline
DIN-SQL & \multirow{5}{*}{Regeneration} & No & Zero-shot \\ \cline{3-4}
MAC-SQL &  & Yes & Zero-shot \\ \cline{3-4}
RED-SQL &  & Yes$^{\dagger}$ & Zero-shot \\ \cline{3-4}
CHESS-SQL &  & Yes & Few-shot \\ \cline{3-4}
OpenSearch-SQL &  & Yes & Few-shot \\
\hline
SafeQL (ours) & Search + Regeneration & Yes & Few-shot \\
\hline
\end{tabular}
\vspace{1mm}
\small{$^{\dagger}$~RED-SQL uses constraint violation reports as extended error feedback.}
\end{table}
\vspace{-6mm}

\subsubsection{Models.}
We use GPT-OSS-120B~\cite{agarwal2025gpt} as the primary model for all methods, as it provides strong performance among open-source LLMs while ensuring transparent and reproducible evaluation across methods.
We also evaluate SafeQL with two complementary model families: 
\begin{itemize}[labelindent=0em,labelsep=0.5em,leftmargin=*,itemsep=0em,itemsep=0em]
\item \textbf{commercial LLMs}—GPT-4o~~\cite{openai2024gpt4o} and GPT-3.5-turbo~\cite{openai2023gpt35}
\item \textbf{open source LLMs}—Llama-3~\cite{meta2024llama3} and Qwen-2.5~\cite{qwen2024}
\end{itemize}
This diversity allows us to assess SafeQL’s robustness and portability across model architectures and training paradigms.

\subsubsection{Hardware and Software Environment.}
All experiments were conducted on a single server equipped with two AMD EPYC 7302 CPUs (3.0 GHz, 16 cores each), 2096 GB RAM (3200 MHz), and two NVIDIA A100 GPUs. 
LLM inference was served using vLLM~\cite{kwon2023efficient}.

\vspace{-2mm}
\subsection{Main results}
\label{sec:evaluation:main results}

We evaluate SafeQL on both the Bird and Spider benchmarks, measuring its impact on execution accuracy, token efficiency, and refinement time.
Across all settings, SafeQL consistently achieves higher accuracy while consuming far fewer LLM tokens than prior regeneration-based refinement methods.

\subsubsection{Bird results}
\label{sec:evaluation:main results:bird}

Table~\ref{tab:safeql_delta} summarizes the results on the Bird full dev benchmark.
SafeQL achieves the best overall performance, improving execution accuracy to 63.3\% with DAIL-SQL in the prompt-based approach and 69.4\% with OpenSearch-SQL in the agent-based approach—gains of +5.8\% and +5.2\%, respectively—while reducing execution errors by up to 87.4\%, showing its robustness across both approaches.
A deeper look at the regeneration-based methods highlights their inherent trade-offs: 
\begin{itemize}[labelindent=0em,labelsep=0.5em,leftmargin=*,itemsep=0em,itemsep=0em]

\item \textbf{DIN-SQL} 
regenerates queries without execution feedback, and this often makes the queries worse, leading to accuracy loss.

\item \textbf{MAC-SQL}
incorporates raw error messages but provides limited semantic guidance, resulting in modest gains.

\item \textbf{RED-SQL}
improves accuracy with constraint-violation feedback, but at the cost of extremely high token usage and latency.

\item \textbf{CHESS-SQL} and \textbf{OpenSearch-SQL}
leverage few-shot examples for stronger accuracy; among them, OpenSearch-SQL is the most balanced, requiring fewer tokens through a short prompt design while maintaining high performance. However, both still rely on full regeneration after every failure.

\end{itemize}

In contrast, SafeQL(hybrid) employs search-based refinement that explores a structured space of error-directed transformations and invokes regeneration only when necessary.
This approach enables SafeQL to reach valid queries with 1.8–15.1$\times$ fewer tokens and 1.9–29.6$\times$ far shorter refinement latency, achieving both higher accuracy and greater efficiency than existing regeneration-based methods.
In addition, SafeQL(search), which relies solely on search-based refinement without issuing new SQL generations, resolves a large fraction of errors with no token overhead.
\textcolor{notice}{
As reflected in $\Delta$Err (\%), search-based refinement already reduces a substantial portion of errors, leaving only 13\%--17\% of erroneous queries to proceed to hybrid fallback, while incurring a moderate runtime overhead of approximately 1.3--1.8$\times$ on average.
}

\vspace{-1mm}
\subsubsection{Spider results}
\label{sec:evaluation:main results:spider}

Table~\ref{tab:safeql_delta_spider} shows the results on the Spider benchmark in the prompt-based setting (DAIL-SQL), where prompting systems already achieve strong performance and agent-based systems offer limited additional benefit.
SafeQL achieves 91.7\% execution accuracy, improving the baseline by +4.6\% and outperforming regeneration-based methods by 2.6--6.4\%.
SafeQL also demonstrates strong efficiency, requiring 5.1–96.5$\times$ fewer tokens and achieving 4.2–34.5$\times$ shorter refinement latency, while reducing execution errors by 77.4\%.
SafeQL(search) delivers most gains without additional tokens, indicating that search-based refinement remains effective and stable even in larger, diverse settings.

\vspace{-2mm}

\subsection{Accuracy analysis}
\label{sec:evaluation:accuracy}

\subsubsection{Model-wise results}

Table~\ref{tab:safeql_modelwise} compares SafeQL across different model families on the Bird mini dev benchmark.
The results show that SafeQL consistently improves execution accuracy regardless of model size or architecture.
Open-source models such as Qwen and Llama exhibit substantial gains of +5–12\%, narrowing the performance gap with larger proprietary models, while commercial LLMs (GPT-3.5-Turbo, GPT-4o) also improve by +4--9\% in both prompt and agent-based settings.

These results indicate that SafeQL’s search-based refinement is model-agnostic: it strengthens weaker models through explicit structural guidance while still providing measurable gains for strong systems.
This consistency across model families highlights its general applicability and robustness across diverse LLM architectures.

\vspace{-2mm}
{\setlength{\textfloatsep}{4pt}
\setlength{\floatsep}{4pt}
\setlength{\intextsep}{4pt}
\setlength{\abovecaptionskip}{2pt}
\setlength{\belowcaptionskip}{2pt}
\setlength{\tabcolsep}{1.0pt}
\setlength{\tabcolsep}{5pt}
\renewcommand{\arraystretch}{0.85}
\begin{table}[b]
  \caption{
  Model-wise comparison of SafeQL performance on the Bird mini dev benchmark, where $\Delta$EX denotes the accuracy gain over the no-refinement baseline.}
  \label{tab:safeql_modelwise}
  \centering
  \begin{tabular}{lcccc}
    \toprule
    \multirow{2}{*}{\textbf{Model}} &
    \multicolumn{2}{c}{\textbf{Prompt-based}} &
    \multicolumn{2}{c}{\textbf{Agent-based}} \\
    \cmidrule(lr){2-3} \cmidrule(lr){4-5}
    & EX (\%) & $\Delta$EX (\%) &
      EX (\%) & $\Delta$EX (\%) \\
    \midrule
    Qwen-2.5-7B     & $44.0$ & $+9.6$ & $47.2$ & $+9.0$ \\
    Qwen-2.5-72B      & $60.2$ & $+5.8$ & $63.4$ & $+5.2$ \\
    \midrule
    Llama-3-8B      & $37.4$ & $+11.6$ & $39.8$ & $+10.2$ \\
    Llama-3.3-70B     & $60.6$ & $+5.6$ & $63.6$ & $+5.4$ \\
    \midrule
    GPT-3.5-Turbo & $52.4$ & $+9.4$ & $53.0$ & $+8.8$ \\
    GPT-4o        & $61.8$ & $+4.8$ & $66.0$ & $+4.6$ \\
    \bottomrule
  \end{tabular}
\end{table}
}

{
\captionsetup{labelfont={bf,color=notice}}
\setlength{\tabcolsep}{2.5pt}
\renewcommand{\arraystretch}{1.0}
\begin{table*}[t]
  \caption{\textcolor{notice}{Comparison on the Bird full dev benchmark using GPT-OSS-120B, showing the absolute execution accuracy (EX), error reduction ($\Delta$Err), additional LLM costs ($\Delta$Tokens), and elapsed time ($\Delta$Time) for accuracy gains ($\Delta$EX) via refinements.}}
  \label{tab:safeql_delta}
  \centering
  \begin{tabular}{lccccccccccc}
    \toprule
    \multirow{2}{*}{\textbf{Refinements}} &
    \multicolumn{5}{c}{\textbf{Prompt-based (DAIL-SQL)}} &
    \multicolumn{5}{c}{\textbf{Agent-based (OpenSearch-SQL)}} \\
    \cmidrule(lr){2-6} \cmidrule(lr){7-11}
    & EX (\%) & $\Delta$EX (\%) & $\Delta$Err (\%) & $\Delta$Tokens ($K$) & $\Delta$Time (s)
    & EX (\%) & $\Delta$EX (\%) & $\Delta$Err (\%) & $\Delta$Tokens ($K$) & $\Delta$Time (s) \\
    \midrule
    No refinement   & $57.5$ & -- & -- & -- & -- & $64.2$ & -- & -- & -- & -- \\
    \midrule
    DIN-SQL~\cite{pourreza_din-sql_2023}         & $57.2$ & $-0.3$ & $4.2$ & $+7{,}748$ & $+5{,}079$
                    & $62.7$ & $-1.6$ & $11.0$ & $+118{,}615$ & $+62{,}104$ \\
    MAC-SQL~\cite{wang2025mac}         & $59.3$ & $+1.8$ & $23.0$ & $+3{,}865$ & $+3{,}505$
                    & $65.6$ & $+1.4$ & $19.7$ & $+66{,}358$ & $+42{,}004$ \\
    RED-SQL~\cite{ren_power_2025}         & $62.9$ & $+5.4$ & $44.2$ & $+19{,}594$ & $+20{,}110$
                    & $66.7$ & $+2.5$ & $46.5$ & $+411{,}469$ & $+474{,}390$ \\
    CHESS-SQL~\cite{talaei_chess_2024}       & $62.6$ & $+5.1$ & $69.1$ & $+6{,}653$ & $+3{,}295$
                    & $68.5$ & $+4.3$ & $80.3$ & $+79{,}934$ & $+35{,}573$ \\
    OpenSearch-SQL~\cite{xie_opensearch-sql_2025}  & $62.0$ & $+4.5$ & $63.0$ & $+3{,}799$ & $+3{,}141$
                    & $68.3$ & $+4.1$ & $69.3$ & $+53{,}143$ & $+30{,}880$ \\
    \midrule
    \textbf{SafeQL (search)} & $\textbf{62.5}$ & $\textbf{+5.0}$ & $\textbf{+70.9}$ & $\textbf{\underline{+0}}$ & $\textbf{\underline{+1{,}249}}$
                    & $\textbf{68.1}$ & $\textbf{+3.9}$ & $\textbf{+70.1}$ & $\textbf{\underline{+0}}$ & $\textbf{\underline{+8{,}656}}$ \\
    \textbf{SafeQL (hybrid)} & $\textbf{\underline{63.3}}$ & $\textbf{\underline{+5.8}}$ & $\textbf{\underline{+83.0}}$ & $\textbf{+1{,}305}$ & $\textbf{+1{,}653}$
                    & $\textbf{\underline{69.4}}$ & $\textbf{\underline{+5.2}}$ & $\textbf{\underline{+87.4}}$ & $\textbf{+28{,}779}$ & $\textbf{+16{,}029}$ \\
    \bottomrule
  \end{tabular}
\vspace{-4mm}
\end{table*}
}

\vspace{-3mm}
{\setlength{\textfloatsep}{4pt}
\setlength{\floatsep}{4pt}
\setlength{\intextsep}{4pt}
\setlength{\abovecaptionskip}{2pt}
\setlength{\belowcaptionskip}{4pt}
\setlength{\tabcolsep}{1.0pt}
\renewcommand{\arraystretch}{1.0}
\begin{table}[hbtp]
  \caption{
  Comparison on the Spider benchmark using GPT-OSS-120B under the DAIL-SQL prompt.}
  \label{tab:safeql_delta_spider}
  \centering
  \begin{tabular}{lccccc}
    \toprule
    \multirow{2}{*}{\textbf{Refinements}} &
    \multicolumn{5}{c}{\textbf{Prompt-based}} \\
    \cmidrule(lr){2-6}
    & EX (\%) & $\Delta$EX (\%) & $\Delta$Err (\%) & $\Delta$Tokens ($K$) & $\Delta$Time (s) \\
    \midrule
    No refinement   & $87.1$ & --  & -- & -- & -- \\
    \midrule
    DIN-SQL         & $85.3$ & $-1.8$ & $-3.2$ & $+1,055$ & $+1194$ \\
    MAC-SQL         & $88.4$ & $+1.3$ & $-22.6$ & $+396$ & $+753$ \\
    RED-SQL         & $88.6$ & $+1.5$ & $-29.0$ & $+7{,}528$ & $+6{,}120$ \\
    CHESS-SQL       & $89.1$ & $+2.0$ & $-45.1$ & $+1{,}279$ & $+1{,}305$ \\
    OpenSearch-SQL  & $89.1$ & $+2.0$ & $-48.7$ & $+465$ & $+1{,}050$ \\
    \midrule
    \textbf{SafeQL (search)} & $\textbf{91.3}$ & $\textbf{+4.2}$ & $\textbf{-68.5}$ & $\textbf{\underline{+0}}$ & $\textbf{\underline{+37}}$ \\
    \textbf{SafeQL (hybrid)} & $\textbf{\underline{91.7}}$ & $\textbf{\underline{+4.6}}$ & $\textbf{\underline{-77.4}}$ & $\textbf{+78}$ & $\textbf{+177}$ \\
    \bottomrule
  \end{tabular}
\vspace{-3mm}
\end{table}
}

\begin{figure}[t]
    \centerline{\includegraphics[width=3.3in]{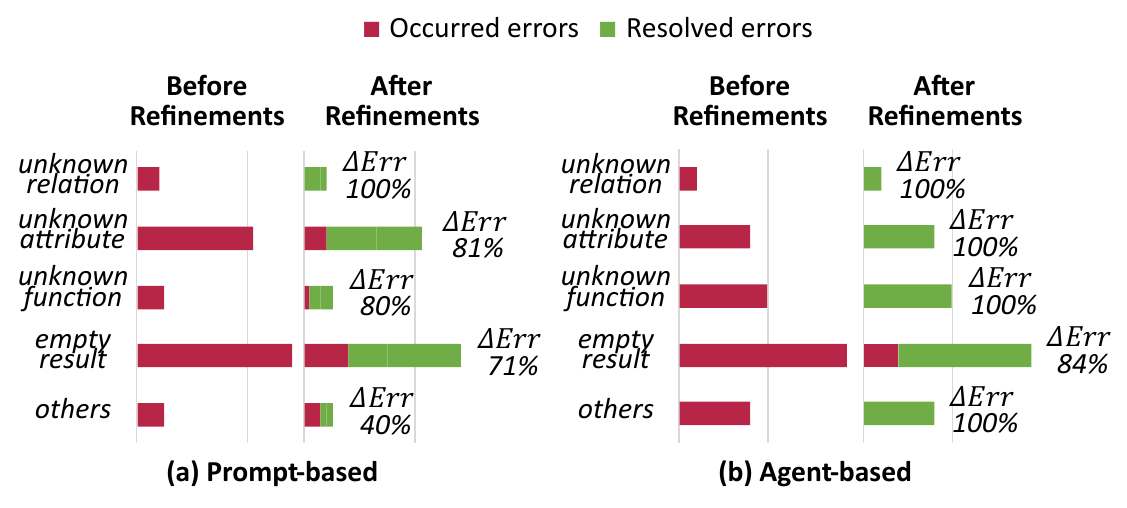}}
    \vspace{-4mm}
    \caption{Error statistics of SafeQL.}
    \label{fig:error-statistics}
\vspace{-4mm}
\end{figure}

\subsubsection{Error analysis}

Figure~\ref{fig:error-statistics} compares the SafeQL's error reduction across the major error categories.
These four representative forms of errors—\textit{unknown relation}, \textit{unknown attribute}, \textit{unknown function}, and \textit{empty result}—cover most of the execution failures observed in practice, and SafeQL consistently reduces all of them in both prompting- and agent-based settings.

In the prompt-based setting, as shown in Figure~\ref{fig:error-statistics}(a), the majority of errors arise from \textit{unknown attribute} or \textit{empty result}.
Because its extraction stage is less sophisticated than that of an agent-based setting, attribute mismatches and overly restrictive predicates appear more frequently.
SafeQL resolves 81\% of \textit{unknown attribute} errors and 71\% of \textit{empty result} cases, demonstrating that search-based refinement substantially improves the reliability of prompt-driven generation.
The remaining failures typically stem from severely misgenerated queries that cannot enter the safe query space, so no sequence of valid refinement steps can ever transform them into an executable query.

In contrast, the agent-based setting, as shown in Figure~\ref{fig:error-statistics}(b) starts from a stronger baseline, as the LLM handles both extraction and generation, leading to a more balanced error distribution.
Even so, SafeQL further eliminates nearly all remaining structural and functional errors—resolving 100\% of \textit{unknown relation}, \textit{attribute}, and \textit{function} errors—and reduces \textit{empty result} errors by 84\%.
These results demonstrate that SafeQL operates orthogonally to the underlying Text-to-SQL systems and that the defined safe query space is not only theoretically sound but also practically effective in capturing nearly all executable refinements in real workloads.

\vspace{-2mm}
\subsection{Efficiency analysis}
\label{sec:evaluation:efficiency}

Figure~\ref{fig:ablation-studies} presents efficiency ablation experiments on the Bird mini dev under the prompt-based setting, isolating the effects of SafeQL’s components:
(1) pruning methods and (2) systematic optimizations.

\begin{itemize}[labelindent=0em,labelsep=0.5em,leftmargin=*,itemsep=0em]

\item \textbf{Effect of pruning.}  
Figure~\ref{fig:ablation-studies}(a) shows that pruning substantially improves search efficiency without any loss of accuracy.
Without pruning, refinement requires 1403\,sec on Bird and 581\,sec on Spider.
Applying either type-based or top-$K$ pruning individually shortens this time, and combining both achieves the best performance—reducing the elapsed time to 647\,sec and 177\,sec, respectively.
These results show that pruning strategies effectively bound the search space while preserving refinement quality.

\item \textbf{Effect of systematic optimizations.}  
Figure~\ref{fig:ablation-studies}(b) evaluates the impact of systematic optimization on refinement latency.
The baseline--lacking caching and indexing--requires over 20{,}000\,sec on Bird and 9{,}000\,sec on Spider.
Introducing caching reduces this time to 1357\,sec and 433\,sec, while indexing alone yields only a modest improvement.
When combined, the total time further decreases to 647\,sec and 177\,sec—amounting to over 30$\times$ overall reduction—again with no degradation in execution accuracy.
Among these optimizations, caching contributes the most by eliminating redundant embedding computations, whereas indexing provides complementary acceleration for similarity lookups.


\end{itemize}

Figure~\ref{fig:parameters} shows how performance changes on the Bird mini dev benchmark, as we vary the two key parameters $\alpha$ and $K$.
We report results in a search-only setting to isolate the effect of these parameters.

\begin{itemize}[labelindent=0em,labelsep=0.5em,leftmargin=*,itemsep=0.2em]

\item \textbf{Effect of $\boldsymbol{\alpha}$.} 
The higher $\alpha$ increases the penalty for structural transformation during refinement, making the search space more compact and substantially reducing latency.  
However, this also introduces a trade-off: mid-range values ($\alpha \approx 0.3$–$0.4$) consistently provide the best execution accuracy, while larger $\alpha$ values prioritize speed over precision.

\item \textbf{Effect of $\boldsymbol{K}$.}  
Raising $K$ allows SafeQL to keep more candidates at each refinement step, effectively widening the search tree.  
This improves accuracy through broader exploration but also increases search time.  
In practice, accuracy saturates once $K \ge 3$, so setting $K$ to at least 3 provides near-optimal performance while keeping runtime manageable.

\end{itemize}



\begin{figure}[hbtp]
    \vspace{-4mm}
    \centerline{\includegraphics[width=3.3in]{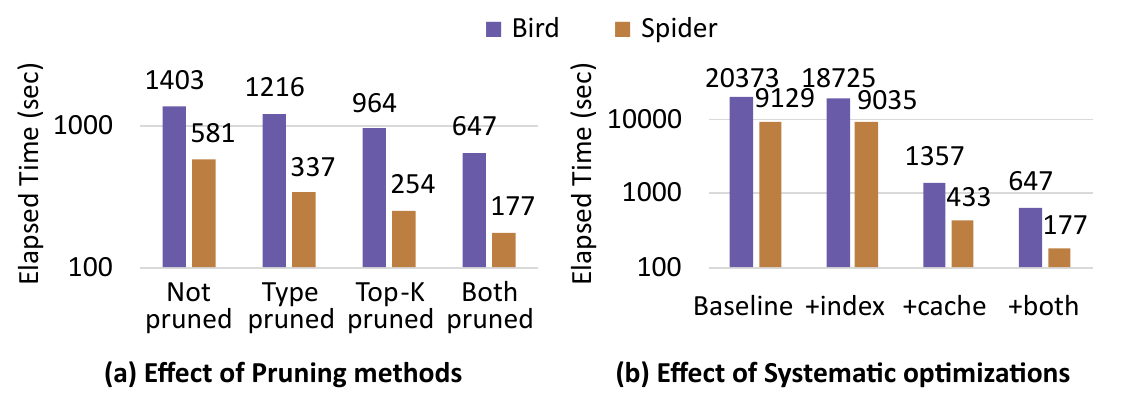}}
    \vspace{-4mm}
    \caption{Efficiency ablation studies.}
    \label{fig:ablation-studies}
    \vspace{-3mm}
\end{figure}

\begin{figure}[hbtp]
    \vspace{-4mm}
    \centerline{\includegraphics[width=3.3in]{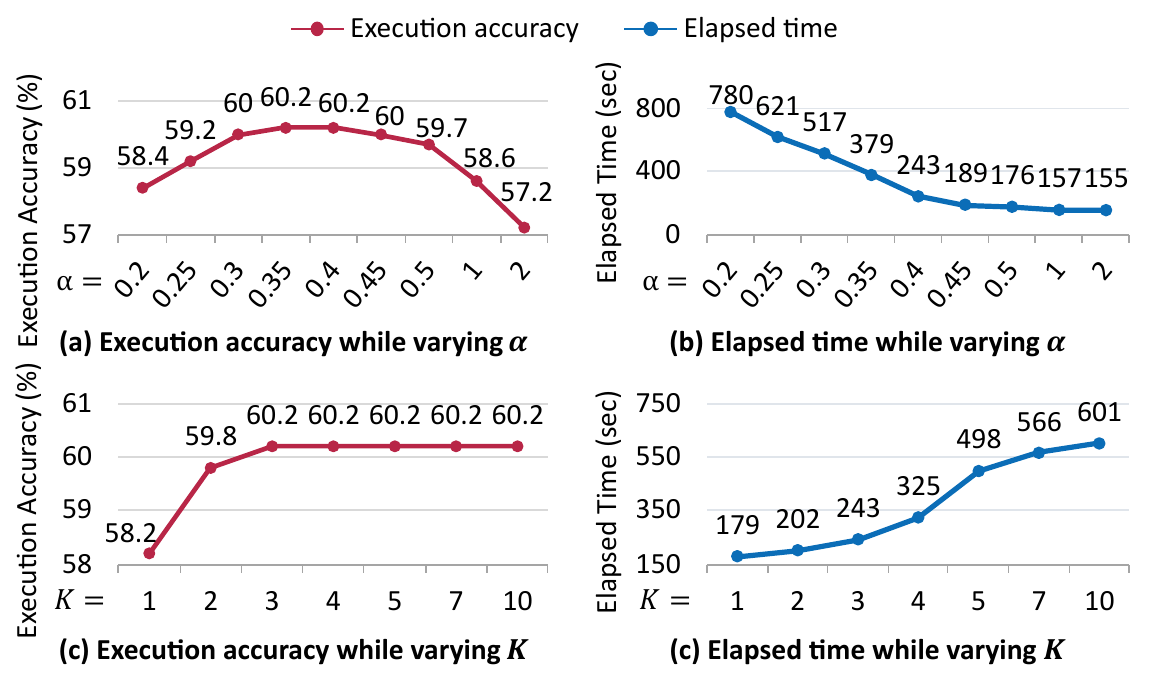}}
    \vspace{-4mm}
    \caption{Effects of the parameters $\alpha$ and $K$.}
    \label{fig:parameters}
    \vspace{-5mm}
\end{figure}

\subsection{System overhead}
\label{sec:evaluation:overhead}

\textcolor{notice}
{
Table~\ref{tab:overhead} reports the system overhead of SafeQL on the Bird dev database; Spider is omitted due to its small size (<1~GB), where both the embedding cache and index are negligible.
The reported primary storage includes all tuples and value embeddings, which are common in Text-to-SQL systems~\cite{liu2025survey} and are not introduced by SafeQL.
SafeQL adds two auxiliary structures—an embedding cache and a vector index—both substantially smaller than the primary data, resulting in minimal storage overhead.
Although index construction is relatively expensive with a single thread (2472 sec), it can be significantly reduced via parallelism (295 sec).
Since construction is performed once upfront and subsequent updates can be handled via incremental strategies~\cite{xu2023spfresh, Yu26Greater}, the overhead remains moderate in practice and can be further reduced with future optimizations.
}

{
\vspace{-2mm}
\captionsetup{labelfont={bf,color=notice}}
\setlength{\tabcolsep}{2.0pt}
\renewcommand{\arraystretch}{1.0}
\begin{table}[h]
  \caption{\textcolor{notice}{SafeQL overhead on the Bird database (T = \# threads).}}
  \vspace*{-2mm}
  \label{tab:overhead}
  \centering
  \begin{tabular}{cccc}
    \toprule
    \multirow{2}{*}{\makecell{\textbf{Primary} \\ \textbf{storage}}} 
    & \multicolumn{3}{c}{\textbf{SafeQL overhead}} \\
    \cmidrule(lr){2-4}
    & \textbf{Cache size} 
    & \textbf{Index size} 
    & \textbf{Index build time (1T / 16T)} \\
    \midrule
    8.04 GB & 686.6 MB & 375.2 MB & 2472 sec / 295 sec \\
    \bottomrule
  \end{tabular}
\vspace{-4mm}
\end{table}
}

%% file: sec_8.tex
\section{Related work}
\label{sec:related_work}
\vspace{-1mm}

\begin{itemize}[labelindent=0em,labelsep=0.5em,leftmargin=*]

\item
\textbf{LLM-based Text-to-SQL} can be broadly categorized into \textit{prompt-based} and \textit{agent-based} approaches.
Prompt-based approaches improve generation quality via prompt design.
DAIL-SQL~\cite{gao_text--sql_2024}, MCS-SQL~\cite{lee_mcs-sql_2025}, PET-SQL~\cite{li2024pet}, and CodexDB~\cite{trummer2022codexdb} enable in-context learning through \textit{few-shot prompting}.
Act-SQL~\cite{zhang2023act} and CoT-SQL~\cite{tai_exploring_2023} enhance reasoning via \textit{Chain-of-Thought}, while DTS-SQL~\cite{pourreza2024dts} and DIN-SQL~\cite{pourreza_din-sql_2023} uses a \textit{decomposition} to breakdown complex questions.
Agent-based approaches coordinate agents to generate queries and maintain consistency.
C3~\cite{dong2023c3}, CHESS~\cite{talaei_chess_2024}, and CSC-SQL~\cite{sheng_csc-sql_2025} select consistent outputs across candidates, OpenSearch-SQL~\cite{xie_opensearch-sql_2025} aligns intermediate states, and Alpha-SQL~\cite{li_alpha-sql_2025} applies Monte Carlo Tree Search over agent actions.

\item 
\textcolor{notice}{
\textbf{Search techniques in Text-to-SQL} are used in decoding during generation, where greedy search selects the most probable token, as in SQLNet~\cite{xu_sqlnet_2017} and Seq2SQL~\cite{zhong_seq2sql_2017}, and beam search explores multiple candidates, as in RAT-SQL~\cite{wang_rat-sql_2020}, EditSQL~\cite{zhang2019editing}, SmBoP~\cite{rubin_smbop_2021}, and ZeroNL2SQL~\cite{gu2023interleaving}.
At a higher level, Alpha-SQL~\cite{li_alpha-sql_2025} applies search via MCTS over agent actions to guide generation.
These approaches perform search during generation, exploring multiple candidates while constructing SQL.
However, while substantially enhancing generation quality and robustness, this also involves maintaining multiple candidates during decoding and generation, introducing additional computational overhead, which is further exacerbated in large models (e.g., LLMs), along with hallucination issues inherent to such models that may compromise query safety.
In contrast, SafeQL introduces search in the refinement stage within the DBMS, leveraging execution feedback to ensure safety while avoiding unnecessary regeneration, enabling faster and more token-efficient refinement.
}

\end{itemize}

\vspace{-4mm}

%% file: sec_9.tex
\section{Conclusion and Discussion}
\label{sec:conclusion}
\vspace{-1mm}

SafeQL introduces a novel search-based perspective on Text-to-SQL refinement, transforming database feedback into structured guidance for executable, intent-preserving queries.
By coupling semantic search with in-database optimization, it achieves both reliability and efficiency, improving execution accuracy by up to 5.8\% while reducing token usage by 15$\times$ over regeneration-based baselines.
Rather than replacing existing generation paradigms, SafeQL complements them through structured, error-guided search process tightly coupled with database semantics.

More broadly, this work relates to efforts to improve the safety of LLMs under formal correctness constraints, where hallucinations often produce invalid or non-executable results.
Recent work in Automatic Program Repair (APR) similarly treats generation as a constraint-guided process integrated with type systems~\cite{zhang2024systematic, mundler2025type, nagy2026chopchop}.
Likewise, SafeQL uses the DBMS analyzer as structured constraints to ground refinement in database semantics and ensure execution-safe query construction, suggesting broader integration of symbolic database reasoning with LLMs for more reliable database interfaces.

%% file: main.bib
@misc{sheng_csc-sql_2025,
  title = {{CSC}-{SQL}: Corrective Self-Consistency in Text-to-SQL via Reinforcement Learning},
  shorttitle = {{CSC}-{SQL}},
  url = {http://arxiv.org/abs/2505.13271},
  doi = {10.48550/arXiv.2505.13271},
  author = {Sheng, Lei and Xu, Shuai-Shuai},
  month = jun,
  year = {2025},
  publisher = {arXiv}
}

@inproceedings{pourreza_chase-sql_2024,
  title = {{CHASE}-{SQL}: Multi-Path Reasoning and Preference Optimized Candidate Selection in Text-to-SQL},
  shorttitle = {{CHASE}-{SQL}},
  url = {https://openreview.net/forum?id=CvGqMD5OtX},
  author = {Pourreza, Mohammadreza and Li, Hailong and Sun, Ruoxi and Chung, Yeounoh and Talaei, Shayan and Kakkar, Gaurav Tarlok and Gan, Yu and Saberi, Amin and Ozcan, Fatma and Arik, Sercan O.},
  month = oct,
  year = {2024}
}

@article{xie_opensearch-sql_2025,
  title = {{OpenSearch}-{SQL}: Enhancing Text-to-SQL with Dynamic Few-shot and Consistency Alignment},
  shorttitle = {{OpenSearch}-{SQL}},
  url = {https://dl.acm.org/doi/10.1145/3725331},
  doi = {10.1145/3725331},
  journal = {Proc. ACM Manag. Data},
  author = {Xie, Xiangjin and Xu, Guangwei and Zhao, Lingyan and Guo, Ruijie},
  volume = {3},
  number = {3},
  pages = {194:1--194:24},
  year = {2025}
}

@inproceedings{li_alpha-sql_2025,
  title = {Alpha-{SQL}: Zero-Shot Text-to-SQL using Monte Carlo Tree Search},
  shorttitle = {Alpha-{SQL}},
  url = {https://openreview.net/forum?id=kGg1ndttmI},
  author = {Li, Boyan and Zhang, Jiayi and Fan, Ju and Xu, Yanwei and Chen, Chong and Tang, Nan and Luo, Yuyu},
  month = jun,
  year = {2025}
}

@misc{talaei_chess_2024,
  title = {{CHESS}: Contextual Harnessing for Efficient SQL Synthesis},
  shorttitle = {{CHESS}},
  url = {https://arxiv.org/abs/2405.16755v3},
  author = {Talaei, Shayan and Pourreza, Mohammadreza and Chang, Yu-Chen and Mirhoseini, Azalia and Saberi, Amin},
  month = may,
  year = {2024}
}

@article{gao_text--sql_2024,
  title = {Text-to-SQL Empowered by Large Language Models: A Benchmark Evaluation},
  shorttitle = {Text-to-SQL Empowered by Large Language Models},
  url = {https://dl.acm.org/doi/10.14778/3641204.3641221},
  doi = {10.14778/3641204.3641221},
  journal = {Proc. VLDB Endow.},
  volume = {17},
  number = {5},
  pages = {1132--1145},
  author = {Gao, Dawei and Wang, Haibin and Li, Yaliang and Sun, Xiuyu and Qian, Yichen and Ding, Bolin and Zhou, Jingren},
  year = {2024}
}

@inproceedings{lee_mcs-sql_2025,
  title = {{MCS}-{SQL}: Leveraging Multiple Prompts and Multiple-Choice Selection For Text-to-SQL Generation},
  shorttitle = {{MCS}-{SQL}},
  url = {https://aclanthology.org/2025.coling-main.24/},
  booktitle = {Proceedings of the 31st International Conference on Computational Linguistics},
  author = {Lee, Dongjun and Park, Choongwon and Kim, Jaehyuk and Park, Heesoo},
  month = jan,
  year = {2025},
  pages = {337--353}
}

@inproceedings{nan_enhancing_2023,
  title = {Enhancing Text-to-SQL Capabilities of Large Language Models: A Study on Prompt Design Strategies},
  url = {https://aclanthology.org/2023.findings-emnlp.996/},
  doi = {10.18653/v1/2023.findings-emnlp.996},
  booktitle = {Findings of the Association for Computational Linguistics: EMNLP 2023},
  author = {Nan, Linyong and Zhao, Yilun and Zou, Weijin and Ri, Narutatsu and Tae, Jaesung and Zhang, Ellen and Cohan, Arman and Radev, Dragomir},
  month = feb,
  year = {2023},
  pages = {14935--14956}
}

@inproceedings{pourreza_din-sql_2023,
  title = {{DIN}-{SQL}: Decomposed In-Context Learning of Text-to-SQL with Self-Correction},
  shorttitle = {{DIN}-{SQL}},
  booktitle = {Proceedings of the 37th International Conference on Neural Information Processing Systems},
  author = {Pourreza, Mohammadreza and Rafiei, Davood},
  year = {2023},
  pages = {36339--36348}
}

@inproceedings{zhou_least--most_2022,
  title = {Least-to-Most Prompting Enables Complex Reasoning in Large Language Models},
  url = {https://openreview.net/forum?id=WZH7099tgfM},
  author = {Zhou, Denny and Schärli, Nathanael and Hou, Le and Wei, Jason and Scales, Nathan and Wang, Xuezhi and Schuurmans, Dale and Cui, Claire and Bousquet, Olivier and Le, Quoc V. and Chi, Ed H.},
  month = sep,
  year = {2022}
}

@inproceedings{tai_exploring_2023,
  title = {Exploring Chain of Thought Style Prompting for Text-to-SQL},
  url = {https://aclanthology.org/2023.emnlp-main.327/},
  doi = {10.18653/v1/2023.emnlp-main.327},
  author = {Tai, Chang-Yu and Chen, Ziru and Zhang, Tianshu and Deng, Xiang and Sun, Huan},
  month = dec,
  year = {2023},
  pages = {5376--5393}
}

@article{wei_chain--thought_2022,
  title = {Chain-of-Thought Prompting Elicits Reasoning in Large Language Models},
  volume = {35},
  journal = {Advances in Neural Information Processing Systems},
  author = {Wei, Jason and Wang, Xuezhi and Schuurmans, Dale and Bosma, Maarten and Ichter, Brian and Xia, Fei and Chi, Ed and Le, Quoc V. and Zhou, Denny},
  month = dec,
  year = {2022},
  pages = {24824--24837}
}

@article{popescu_towards_nodate,
  title = {Towards a Theory of Natural Language Interfaces to Databases},
  author = {Popescu, Ana-Maria and Etzioni, Oren and Kautz, Henry}
}

@inproceedings{li_nalir_2014,
  title = {{NaLIR}: An Interactive Natural Language Interface for Querying Relational Databases},
  shorttitle = {{NaLIR}},
  url = {https://dl.acm.org/doi/10.1145/2588555.2594519},
  doi = {10.1145/2588555.2594519},
  booktitle = {Proceedings of the 2014 ACM SIGMOD International Conference on Management of Data},
  author = {Li, Fei and Jagadish, Hosagrahar V.},
  month = jun,
  year = {2014},
  pages = {709--712}
}

@article{yaghmazadeh_sqlizer_2017,
  title = {{SQLizer}: Query Synthesis from Natural Language},
  shorttitle = {{SQLizer}},
  url = {https://dl.acm.org/doi/10.1145/3133887},
  doi = {10.1145/3133887},
  journal = {Proc. ACM Program. Lang.},
  author = {Yaghmazadeh, Navid and Wang, Yuepeng and Dillig, Isil and Dillig, Thomas},
  volume = {1},
  number = {OOPSLA},
  year = {2017},
  pages = {1--26}
}

@article{saha_athena_2016,
  title = {{ATHENA}: An Ontology-Driven System for Natural Language Querying over Relational Data Stores},
  shorttitle = {{ATHENA}},
  url = {https://dl.acm.org/doi/10.14778/2994509.2994536},
  doi = {10.14778/2994509.2994536},
  journal = {Proc. VLDB Endow.},
  volume = {9},
  number = {12},
  pages = {1209--1220},
  author = {Saha, Diptikalyan and Floratou, Avrilia and Sankaranarayanan, Karthik and Minhas, Umar Farooq and Mittal, Ashish R. and Özcan, Fatma},
  year = {2016}
}

@misc{zhong_seq2sql_2017,
  title = {{Seq2SQL}: Generating Structured Queries from Natural Language using Reinforcement Learning},
  shorttitle = {{Seq2SQL}},
  url = {http://arxiv.org/abs/1709.00103},
  doi = {10.48550/arXiv.1709.00103},
  author = {Zhong, Victor and Xiong, Caiming and Socher, Richard},
  month = nov,
  year = {2017}
}

@misc{xu_sqlnet_2017,
  title = {{SQLNet}: Generating Structured Queries From Natural Language Without Reinforcement Learning},
  shorttitle = {{SQLNet}},
  url = {http://arxiv.org/abs/1711.04436},
  doi = {10.48550/arXiv.1711.04436},
  author = {Xu, Xiaojun and Liu, Chang and Song, Dawn},
  month = nov,
  year = {2017}
}

@misc{guo_towards_2019,
  title = {Towards Complex Text-to-SQL in Cross-Domain Database with Intermediate Representation},
  url = {http://arxiv.org/abs/1905.08205},
  doi = {10.48550/arXiv.1905.08205},
  author = {Guo, Jiaqi and Zhan, Zecheng and Gao, Yan and Xiao, Yan and Lou, Jian-Guang and Liu, Ting and Zhang, Dongmei},
  month = may,
  year = {2019}
}

@inproceedings{wang_rat-sql_2020,
  title = {{RAT}-{SQL}: Relation-Aware Schema Encoding and Linking for Text-to-SQL Parsers},
  shorttitle = {{RAT}-{SQL}},
  url = {https://aclanthology.org/2020.acl-main.677/},
  doi = {10.18653/v1/2020.acl-main.677},
  booktitle = {Proceedings of the 58th Annual Meeting of the Association for Computational Linguistics},
  author = {Wang, Bailin and Shin, Richard and Liu, Xiaodong and Polozov, Oleksandr and Richardson, Matthew},
  month = jul,
  year = {2020},
  pages = {7567--7578}
}

@inproceedings{rubin_smbop_2021,
  title = {{SmBoP}: Semi-autoregressive Bottom-up Semantic Parsing},
  shorttitle = {{SmBoP}},
  url = {https://aclanthology.org/2021.naacl-main.29/},
  doi = {10.18653/v1/2021.naacl-main.29},
  author = {Rubin, Ohad and Berant, Jonathan},
  booktitle = {Proceedings of the 2021 Conference of the North American Chapter of the Association for Computational Linguistics: Human Language Technologies},
  month = jun,
  year = {2021},
  pages = {311--324}
}

@inproceedings{scholak_picard_2021,
  title = {{PICARD}: Parsing Incrementally for Constrained Auto-Regressive Decoding from Language Models},
  shorttitle = {{PICARD}},
  url = {https://aclanthology.org/2021.emnlp-main.779/},
  doi = {10.18653/v1/2021.emnlp-main.779},
  author = {Scholak, Torsten and Schucher, Nathan and Bahdanau, Dzmitry},
  month = jan,
  year = {2021},
  pages = {9895--9901}
}

@misc{lin_bridging_2020,
  title = {Bridging Textual and Tabular Data for Cross-Domain Text-to-SQL Semantic Parsing},
  url = {http://arxiv.org/abs/2012.12627},
  doi = {10.48550/arXiv.2012.12627},
  author = {Lin, Xi Victoria and Socher, Richard and Xiong, Caiming},
  month = dec,
  year = {2020}
}

@inproceedings{brown_language_2020,
  title = {Language Models are Few-Shot Learners},
  url = {https://proceedings.neurips.cc/paper_files/paper/2020/hash/1457c0d6bfcb4967418bfb8ac142f64a-Abstract.html},
  booktitle = {Advances in Neural Information Processing Systems},
  author = {Brown, Tom and Mann, Benjamin and Ryder, Nick and Subbiah, Melanie and Kaplan, Jared D. and Dhariwal, Prafulla and Neelakantan, Arvind and Shyam, Pranav and Sastry, Girish and Askell, Amanda and Agarwal, Sandhini and Herbert-Voss, Ariel and Krueger, Gretchen and Henighan, Tom and Child, Rewon and Ramesh, Aditya and Ziegler, Daniel and Wu, Jeffrey and Winter, Clemens and Hesse, Chris and Chen, Mark and Sigler, Eric and Litwin, Mateusz and Gray, Scott and Chess, Benjamin and Clark, Jack and Berner, Christopher and McCandlish, Sam and Radford, Alec and Sutskever, Ilya and Amodei, Dario},
  year = {2020},
  pages = {1877--1901}
}

@article{ren_power_2025,
  title = {The Power of Constraints in Natural Language to SQL Translation},
  journal = {Proc. VLDB Endow.},
  volume = {18},
  number = {7},
  pages = {2097--2111},
  doi = {10.14778/3734839.3734847},
  author = {Ren, Tonghui and Ke, Chen and Fan, Yuankai and Jing, Yinan and He, Zhenying and Zhang, Kai and Wang, X. Sean},
  month = mar,
  year = {2025}
}

@article{zhang_simple_1989,
  title = {Simple Fast Algorithms for the Editing Distance between Trees and Related Problems},
  journal = {SIAM Journal on Computing},
  volume = {18},
  number = {6},
  pages = {1245--1262},
  author = {Zhang, Kaizhong and Shasha, Dennis},
  month = dec,
  year = {1989}
}

@misc{reimers_sentence-bert_2019,
  title = {Sentence-BERT: Sentence Embeddings using Siamese BERT-Networks},
  url = {http://arxiv.org/abs/1908.10084},
  doi = {10.48550/arXiv.1908.10084},
  author = {Reimers, Nils and Gurevych, Iryna},
  month = aug,
  year = {2019}
}

@inproceedings{santhanam_colbertv2_2022,
  title = {{ColBERTv2}: Effective and Efficient Retrieval via Lightweight Late Interaction},
  shorttitle = {{ColBERTv2}},
  url = {https://aclanthology.org/2022.naacl-main.272/},
  doi = {10.18653/v1/2022.naacl-main.272},
  author = {Santhanam, Keshav and Khattab, Omar and Saad-Falcon, Jon and Potts, Christopher and Zaharia, Matei},
  booktitle = {Proceedings of the 2022 Conference of the North American Chapter of the Association for Computational Linguistics: Human Language Technologies},
  month = jul,
  year = {2022},
  pages = {3715--3734}
}

@inproceedings{formal_splade_2021,
  title = {{SPLADE}: Sparse Lexical and Expansion Model for First Stage Ranking},
  shorttitle = {{SPLADE}},
  url = {https://dl.acm.org/doi/10.1145/3404835.3463098},
  doi = {10.1145/3404835.3463098},
  booktitle = {Proceedings of the 44th International ACM SIGIR Conference on Research and Development in Information Retrieval},
  author = {Formal, Thibault and Piwowarski, Benjamin and Clinchant, Stéphane},
  month = jul,
  year = {2021},
  pages = {2288--2292}
}

@inproceedings{yu-etal-2018-spider,
  title = {{Spider}: A Large-Scale Human-Labeled Dataset for Complex and Cross-Domain Semantic Parsing and Text-to-{SQL} Task},
  author = {Yu, Tao and Zhang, Rui and Yang, Kai and Yasunaga, Michihiro and Wang, Dongxu and Li, Zifan and Ma, James and Li, Irene and Yao, Qingning and Roman, Shanelle and Zhang, Zilin and Radev, Dragomir},
  editor = {Riloff, Ellen and Chiang, David and Hockenmaier, Julia and Tsujii, Jun'ichi},
  booktitle = {Proceedings of the 2018 Conference on Empirical Methods in Natural Language Processing},
  month = {oct--nov},
  year = {2018},
  address = {Brussels, Belgium},
  publisher = {Association for Computational Linguistics},
  url = {https://aclanthology.org/D18-1425/},
  doi = {10.18653/v1/D18-1425},
  pages = {3911--3921}
}

@article{li2024can,
  title = {Can LLM Already Serve as a Database Interface? A Big Bench for Large-Scale Database Grounded Text-to-SQLs},
  author = {Li, Jinyang and Hui, Binyuan and Qu, Ge and Yang, Jiaxi and Li, Binhua and Li, Bowen and Wang, Bailin and Qin, Bowen and Geng, Ruiying and Huo, Nan and others},
  journal = {Advances in Neural Information Processing Systems},
  volume = {36},
  year = {2024}
}

@article{trummer2022codexdb,
  title={CodexDB: Synthesizing code for query processing from natural language instructions using GPT-3 Codex},
  author={Trummer, Immanuel},
  journal={Proceedings of the VLDB Endowment},
  volume={15},
  number={11},
  pages={2921--2928},
  year={2022},
  publisher={VLDB Endowment}
}

@article{li2024pet,
  title={Pet-sql: A prompt-enhanced two-stage Text-to-SQL framework with cross-consistency},
  author={Li, Zhishuai and Wang, Xiang and Zhao, Jingjing and Yang, Sun and Du, Guoqing and Hu, Xiaoru and Zhang, Bin and Ye, Yuxiao and Li, Ziyue and Zhao, Rui and others},
  journal={CoRR},
  year={2024}
}

@inproceedings{zhang2023act,
  title={ACT-SQL: In-Context Learning for Text-to-SQL with Automatically-Generated Chain-of-Thought},
  author={Zhang, Hanchong and Cao, Ruisheng and Chen, Lu and Xu, Hongshen and Yu, Kai},
  booktitle={The 2023 Conference on Empirical Methods in Natural Language Processing}
}

@inproceedings{pourreza2024dts,
  title={DTS-SQL: Decomposed Text-to-SQL with Small Large Language Models},
  author={Pourreza, Mohammadreza and Rafiei, Davood},
  booktitle={Findings of the Association for Computational Linguistics: EMNLP 2024},
  pages={8212--8220},
  year={2024}
}

@inproceedings{wang2025mac,
  title={Mac-sql: A multi-agent collaborative framework for text-to-sql},
  author={Wang, Bing and Ren, Changyu and Yang, Jian and Liang, Xinnian and Bai, Jiaqi and Chai, Linzheng and Yan, Zhao and Zhang, Qian-Wen and Yin, Di and Sun, Xing and others},
  booktitle={Proceedings of the 31st International Conference on Computational Linguistics},
  pages={540--557},
  year={2025}
}

@article{dong2023c3,
  title={C3: Zero-shot text-to-sql with chatgpt},
  author={Dong, Xuemei and Zhang, Chao and Ge, Yuhang and Mao, Yuren and Gao, Yunjun and Lin, Jinshu and Lou, Dongfang and others},
  journal={arXiv preprint arXiv:2307.07306},
  year={2023}
}

@inproceedings{zhang2019editing,
  title={Editing-based SQL query generation for cross-domain context-dependent questions},
  author={Zhang, Rui and Yu, Tao and Er, Heyang and Shim, Sungrok and Xue, Eric and Lin, Xi Victoria and Shi, Tianze and Xiong, Caiming and Socher, Richard and Radev, Dragomir},
  booktitle={Proceedings of the 2019 Conference on Empirical Methods in Natural Language Processing and the 9th International Joint Conference on Natural Language Processing (EMNLP-IJCNLP)},
  pages={5338--5349},
  year={2019}
}

@article{gu2023interleaving,
  title={Interleaving pre-trained language models and large language models for zero-shot nl2sql generation},
  author={Gu, Zihui and Fan, Ju and Tang, Nan and Zhang, Songyue and Zhang, Yuxin and Chen, Zui and Cao, Lei and Li, Guoliang and Madden, Sam and Du, Xiaoyong},
  journal={arXiv preprint arXiv:2306.08891},
  year={2023}
}

@inproceedings{xu2023spfresh,
  title={Spfresh: Incremental in-place update for billion-scale vector search},
  author={Xu, Yuming and Liang, Hengyu and Li, Jin and Xu, Shuotao and Chen, Qi and Zhang, Qianxi and Li, Cheng and Yang, Ziyue and Yang, Fan and Yang, Yuqing and others},
  booktitle={Proceedings of the 29th Symposium on Operating Systems Principles},
  pages={545--561},
  year={2023}
}

@article{Yu26Greater,
author = {Yu, Song and Lin, Shengyuan and Gong, Shufeng and Xie, Yongqing and Liu, Ruicheng and Zhou, Yijie and Sun, Ji and Zhang, Yanfeng and Li, Guoliang and Yu, Ge},
title = {A Topology-Aware Localized Update Strategy for Graph-Based ANN Index},
year = {2026},
issue_date = {November 2025},
publisher = {VLDB Endowment},
volume = {19},
number = {3},
issn = {2150-8097},
journal = {Proc. VLDB Endow.},
pages = {495–508},
numpages = {14}
}

@article{zhang2024systematic,
  title={A systematic literature review on large language models for automated program repair},
  author={Zhang, Quanjun and Fang, Chunrong and Xie, Yang and Ma, YuXiang and Sun, Weisong and Yang, Yun and Chen, Zhenyu},
  journal={ACM Transactions on Software Engineering and Methodology},
  year={2024},
  publisher={ACM New York, NY}
}

@article{mundler2025type,
  title={Type-constrained code generation with language models},
  author={M{\"u}ndler, Niels and He, Jingxuan and Wang, Hao and Sen, Koushik and Song, Dawn and Vechev, Martin},
  journal={Proceedings of the ACM on Programming Languages},
  volume={9},
  number={PLDI},
  pages={601--626},
  year={2025},
  publisher={ACM New York, NY, USA}
}

@article{nagy2026chopchop,
  title={ChopChop: A Programmable Framework for Semantically Constraining the Output of Language Models},
  author={Nagy, Shaan and Zhou, Timothy and Polikarpova, Nadia and D'Antoni, Loris},
  journal={Proceedings of the ACM on Programming Languages},
  volume={10},
  number={POPL},
  pages={1905--1932},
  year={2026},
  publisher={ACM New York, NY, USA}
}

@inproceedings{zelle1996learning,
  title={Learning to parse database queries using inductive logic programming},
  author={Zelle, John M and Mooney, Raymond J},
  booktitle={Proceedings of the national conference on artificial intelligence},
  pages={1050--1055},
  year={1996}
}

@inproceedings{iacob2020neural,
  title={Neural approaches for natural language interfaces to databases: A survey},
  author={Iacob, Radu Cristian Alexandru and Brad, Florin and Apostol, Elena-Simona and Truic{\u{a}}, Ciprian-Octavian and Hosu, Ionel Alexandru and Rebedea, Traian},
  booktitle={proceedings of the 28th International Conference on Computational Linguistics},
  pages={381--395},
  year={2020}
}

@article{shen2025study,
  title={A Study of In-Context-Learning-Based Text-to-SQL Errors},
  author={Shen, Jiawei and Wan, Chengcheng and Qiao, Ruoyi and Zou, Jiazhen and Xu, Hang and Shao, Yuchen and Zhang, Yueling and Miao, Weikai and Pu, Geguang},
  journal={arXiv preprint arXiv:2501.09310},
  year={2025}
}

@article{narendra1977branch,
  title={A branch and bound algorithm for feature subset selection},
  author={Narendra and Fukunaga},
  journal={IEEE Transactions on computers},
  volume={100},
  number={9},
  pages={917--922},
  year={1977},
  publisher={IEEE}
}

@article{ning2024insights,
  title={Insights into natural language database query errors: From attention misalignment to user handling strategies},
  author={Ning, Zheng and Tian, Yuan and Zhang, Zheng and Zhang, Tianyi and Li, Toby Jia-Jun},
  journal={ACM Transactions on Interactive Intelligent Systems},
  volume={14},
  number={4},
  pages={1--32},
  year={2024},
  publisher={ACM New York, NY}
}

@inproceedings{chen2023text,
  title={Text-to-SQL Error Correction with Language Models of Code},
  author={Chen, Ziru and Chen, Shijie and White, Michael and Mooney, Raymond and Payani, Ali and Srinivasa, Jayanth and Su, Yu and Sun, Huan},
  booktitle={Proceedings of the 61st Annual Meeting of the Association for Computational Linguistics (Volume 2: Short Papers)},
  pages={1359--1372},
  year={2023}
}

@article{kim2020natural,
  title={Natural language to SQL: Where are we today?},
  author={Kim, Hyeonji and So, Byeong-Hoon and Han, Wook-Shin and Lee, Hongrae},
  journal={Proceedings of the VLDB Endowment},
  volume={13},
  number={10},
  pages={1737--1750},
  year={2020},
  publisher={VLDB Endowment}
}

@article{katsogiannis2023survey,
  title={A survey on deep learning approaches for text-to-SQL},
  author={Katsogiannis-Meimarakis, George and Koutrika, Georgia},
  journal={The VLDB Journal},
  volume={32},
  number={4},
  pages={905--936},
  year={2023},
  publisher={Springer}
}

@article{liu2025survey,
  title={A survey of text-to-sql in the era of llms: Where are we, and where are we going?},
  author={Liu, Xinyu and Shen, Shuyu and Li, Boyan and Ma, Peixian and Jiang, Runzhi and Zhang, Yuxin and Fan, Ju and Li, Guoliang and Tang, Nan and Luo, Yuyu},
  journal={IEEE Transactions on Knowledge and Data Engineering},
  year={2025},
  publisher={IEEE}
}

@article{hong2003introduction,
  title={Introduction to Programming Languages},
  author={Hong, Jaemin and Ryu, Sukyoung},
  year={2003}
}

@article{agarwal2025gpt,
  title={gpt-oss-120b \& gpt-oss-20b model card},
  author={Agarwal, Sandhini and Ahmad, Lama and Ai, Jason and Altman, Sam and Applebaum, Andy and Arbus, Edwin and Arora, Rahul K and Bai, Yu and Baker, Bowen and Bao, Haiming and others},
  journal={arXiv preprint arXiv:2508.10925},
  year={2025}
}

@misc{openai2024gpt4o,
  title     = {GPT-4o System Card},
  author    = {OpenAI},
  year      = {2024},
  url       = {https://openai.com},
}

@misc{openai2023gpt35,
  title     = {GPT-3.5 Technical Report},
  author    = {OpenAI},
  year      = {2023},
  url       = {https://openai.com},
}

@misc{meta2024llama3,
  title     = {The Llama 3 Herd of Models},
  author    = {Meta AI},
  year      = {2024},
  url       = {https://ai.meta.com/llama/},
}

@misc{qwen2024,
  title     = {Qwen2 Technical Report},
  author    = {Qwen Team},
  year      = {2024},
  url       = {https://qwen.ai},
}

@inproceedings{kwon2023efficient,
  title={Efficient memory management for large language model serving with pagedattention},
  author={Kwon, Woosuk and Li, Zhuohan and Zhuang, Siyuan and Sheng, Ying and Zheng, Lianmin and Yu, Cody Hao and Gonzalez, Joseph and Zhang, Hao and Stoica, Ion},
  booktitle={Proceedings of the 29th symposium on operating systems principles},
  pages={611--626},
  year={2023}
}

@manual{postgresql,
  title        = {PostgreSQL: The World's Most Advanced Open Source Relational Database},
  author       = {{The PostgreSQL Global Development Group}},
  year         = {2024},
  url          = {https://www.postgresql.org/},
  note         = {Version 17}
}

@software{fastembed,
  title        = {FastEmbed: Lightweight and Efficient Text Embedding Library},
  author       = {Qdrant Team},
  year         = {2024},
  url          = {https://github.com/qdrant/fastembed},
  note         = {Python library for fast text embedding inference}
}

@misc{bge-base-en-v1.5,
  title        = {BGE-base-en-v1.5: English Text Embedding Model},
  author       = {Xiao, Jun and Wang, Chenglin and Huang, Yuxin and others},
  year         = {2024},
  howpublished = {\url{https://huggingface.co/BAAI/bge-base-en-v1.5}},
  note         = {Beijing Academy of Artificial Intelligence (BAAI)}
}

@software{pgvecto.rs,
  title        = {pgvecto.rs: High-performance Vector Search Extension for PostgreSQL},
  author       = {TensorChord Team},
  year         = {2024},
  url          = {https://github.com/tensorchord/pgvecto.rs},
  note         = {Rust-based vector indexing extension for PostgreSQL}
}
